\documentclass[11pt]{article}

\usepackage{graphicx}
\usepackage{a4}
\usepackage{amsmath,natbib}
\usepackage{parskip}
\usepackage{amsfonts} 
\usepackage{enumerate}
\usepackage{rotating}
\usepackage{natbib}
\usepackage{multirow}
\usepackage{subfigure}
\usepackage{texab}
\usepackage{m-floating}
\usepackage[toc,page]{appendix}
\usepackage{amsthm}

\newcommand{\calN}{\mathcal{N}}
\newcommand{\Prob}{\mathbb{P}}
\newcommand{\EE}{\mathbb{E}}
\newcommand{\PP}{\mathbb{P}}
\newcommand{\bx}{{\bf X}}

\newtheorem{theorem}{Theorem}
\newtheorem{proposition}{Proposition}
\newtheorem{lemma}{Lemma}
\newtheorem{remark}{Remark}
\newtheorem{corollary}{Corollary}
\newtheorem{definition}{Definition}

\begin{document}

\title{$\ell_1$-Penalization for Mixture Regression Models}

\author{Nicolas St\"adler, Peter B\"uhlmann and Sara van de Geer\\
\small Seminar for Statistics, ETH Zurich\\[-0.8ex]
\small \small CH-8092 Zurich, Switzerland.\\
\small \texttt{staedler@stat.math.ethz.ch}}

\date{} 

\maketitle

\begin{abstract}
We consider a finite mixture of regressions (FMR) model for high-dimensional
inhomogeneous data where the number of covariates may be much
larger than sample size. We propose an $\ell_1$-penalized maximum likelihood
estimator in an appropriate parameterization. This kind of estimation
belongs to a class of problems where optimization and theory for
non-convex functions is needed. This distinguishes itself very clearly from
high-dimensional estimation with convex loss- or objective functions, as
for example with the Lasso in linear or generalized linear models. Mixture
models represent a prime and important example where non-convexity arises.  

For FMR models, we develop an efficient
EM algorithm for numerical optimization with provable convergence
properties. Our penalized estimator is
numerically better posed (e.g., boundedness of the 
criterion function) than unpenalized maximum likelihood estimation, and it
allows for effective statistical regularization including variable
selection. We also present some asymptotic theory and oracle inequalities:
due to non-convexity of the negative log-likelihood function, different
mathematical arguments are needed than for problems with convex
losses. Finally, we apply 
the new method to both simulated and real data. \vspace{0.5cm}
\vspace{0.5cm}\\
{\bf This is the author’s version of the work (published as a discussion paper in TEST,  2010, Volume 19,  209­-285). The final publication is available at www.springerlink.com.} 
\end{abstract}

\section{Introduction}\label{sec.intro}
In applied statistics, tremendous number of applications deal with
relating a random response variable $Y$ to a set of explanatory variables or
covariates $X = (X^{(1)},\ldots ,X^{(p)})$ through a regression-type model.
The homogeneity assumption that the regression coefficients are the same
for different observations $(Y_1,X_1),\ldots,(Y_n,X_n)$ is often
inadequate. Parameters may change for different subgroups of
observations. Such heterogeneity can be modeled with a Finite Mixture
of Regressions (FMR) model. Especially with high-dimensional data, where
the number of covariates $p$ is much larger than sample size $n$, the
homogeneity assumption seems rather restrictive: at least a fraction of
covariates may exhibit a different influence on the response among various
observations (i.e., sub-populations). Hence, addressing the issue of
heterogeneity in high-dimensional data is important in many
practical applications. We will empirically demonstrate with real data in
Section \ref{subsec.riboflavin} that substantial prediction improvements
are possible by 
incorporating a heterogeneity structure to the model.   

We propose here an $\ell_1$-penalized method, i.e., a Lasso-type estimator
\citep{tibshirani96regression}, for estimating a high-dimensional Finite
Mixture of Regressions (FMR) model where $p \gg n$. Our procedure is related
to the proposal in \cite{khalili}. However, we argue that a different
parameterization leads to more efficient 
computation in high-dimensional optimization for which we prove numerical
convergence properties. Our algorithm can easily handle problems where $p$
is in the thousands.  
Furthermore, regarding statistical performance, we present an oracle inequality which includes the
setting where $p \gg n$: this is very different from \citet{khalili} who use
fixed $p$ asymptotics in the low-dimensional framework. Our theory for
deriving oracle inequalities in the presence of non-convex loss functions, as
the negative 
log-likelihood in a mixture model is non-convex, is rather
non-standard but sufficiently general to cover other cases than FMR
models. 

From a more general point of view, we show in this paper that
high-dimensional estimation problems with non-convex loss 
functions can be addressed with high computational efficiency and good
statistical accuracy. Regarding the computation, we develop a rather
generic block
coordinate descent generalized EM algorithm which is surprisingly fast
even for large $p$. Progress in efficient gradient descent methods build on
various 
developments by \cite{tseng} and \cite{tseng08}, and their use for solving
Lasso-type convex problems has been worked out by, e.g., \cite{Fu98},
\cite{friedman07fastlasso}, \cite{meier06group} and
\cite{friedmanetal08}. We present in Section \ref{subsec.comptiming}
some computation times for the more involved case with non-convex objective
function using a block coordinate descent generalized
EM algorithm. Regarding statistical theory, almost all results for
high-dimensional Lasso-type problems have been developed for convex loss
functions, e.g., the squared error in a 
Gaussian regression \citep{greenshtein03persistency,
  meinshausen04consistent, zhao05model, Bunea:07, zhanghua08, meyu09,
  wainwright09, bickel07dantzig, cai09, canpan07, tzhang09} or the negative 
log-likelihood in a generalized linear model \citep{geer06high}. We present
a non-trivial modification of the mathematical 
analysis of $\ell_1$-penalized estimation with convex loss to non-convex
but smooth likelihood problems.  

When estimation is defined via optimization of a non-convex
objective function, there is a major gap between the actual computation
and the procedure studied in theory. The computation is typically
guaranteed to find a local optimum of 
the objective function only, whereas the theory is usually about the
estimator defined by a global optimum. Particularly in high-dimensional
problems, it is difficult 
to compute a global optimum and it would be desirable to have some
theoretical properties of estimators arising from local optima. We do not
provide an answer to this difficult issue in this thesis. The beauty of, e.g., the Lasso or the Dantzig selector \citep{cantao07} in high-dimensional
problems is the provable correctness or optimality of the estimator which
is actually computed. A challenge for future research is to establish such
provable correctness of estimators involving non-convex objective
functions. A noticeable exception is presented in \cite{chzhang09} for linear
models, where some theory is derived for an estimator based on a local
optimum of a non-convex optimization criterion. 

The rest of this article is mainly focusing on Finite Mixture of
Regressions (FMR) models. Some theory for high-dimensional estimation with
non-convex loss functions is presented in Section \ref{subsec.ashighdim} for more
general settings than FMR models. The further organization of the paper is as
follows: Section \ref{sec.FMR}
describes the FMR model with an appropriate
parameterization, Section \ref{sec:penregr} introduces $\ell_1$-penalized
maximum-likelihood estimation, Sections \ref{sec.asympt} and
\ref{subsec.ashighdim} present mathematical theory for the low- and
high-dimensional case, Section \ref{sec.optim} develops some efficient
generalized EM algorithm and describes its numerical convergence properties,
and Section \ref{sec.numeric} 
reports on simulations, real data analysis and computational timings.

\section{Finite mixture of Gaussian regressions model}\label{sec.FMR}
Our primary focus is on the following mixture model involving Gaussian
components: 
\begin{eqnarray}\label{mod.mix}
& &Y_i|X_i\ \mbox{independent for}\ i=1,\ldots ,n,\nonumber\\
& &Y_i|X_i = x \sim f_{\xi}(y|x)dy\ \mbox{for}\ i=1,\ldots ,n,\nonumber\\
& &f_{\xi}(y|x) = \sum_{r=1}^{k}\pi_{r} \frac{1}{\sqrt{2 \pi} \sigma_r}
\exp(-\frac{(y - x^T \beta_r)^2}{2 \sigma_r^2}),\\
& &\xi=(\beta_{1},\ldots,\beta_{k},\sigma_{1},\ldots,\sigma_{k},
  \pi_{1},\ldots,\pi_{k-1})\in \R^{kp}\times\R_{>0}^{k}\times\Pi,\nonumber\\
& &\Pi = \{\pi; \pi_r > 0\ \mbox{for}\ r=1,\ldots ,k-1\ \mbox{and}\
\sum_{r=1}^{k-1} \pi_r < 1\}.\nonumber
\end{eqnarray} 
Thereby, $X_i \in \R^p$ are fixed or random covariates, $Y_i \in \R$ is a univariate response
variable and $\xi = (\beta_{1},\ldots,\beta_{k},\sigma_{1},\ldots,\sigma_{k},
  \pi_{1},\ldots,\pi_{k-1})$ denotes the $(p+2) \cdot k -1 $
free parameters and $\pi_{k}$ is given by $\pi_{k}=1-\sum_{r=1}^{k-1}\pi_{r}$. The model in
(\ref{mod.mix}) is a mixture of Gaussian regressions, where every component 
$r$ has its individual vector of regression coefficients $\beta_r$ and
error variances $\sigma^2_r$. We are particularly interested in the case
where $p \gg n$.  

\subsection{Reparameterized mixture of regressions model}\label{subsec.reparmix}
We prefer to work with a reparameterized version of model
(\ref{mod.mix}) whose penalized maximum likelihood estimator is
scale-invariant and easier to compute. The computational aspect will be
discussed in greater 
detail in Sections \ref{subsec.reparlin} and \ref{sec.optim}. Define new
parameters 
\begin{eqnarray*}
\phi_r = \beta_r/\sigma_r,\ \ \rho_r = \sigma_r^{-1},\ \ r=1,\ldots ,k.
\end{eqnarray*}
This yields a one-to-one mapping from $\xi$ in (\ref{mod.mix}) to a new
parameter vector $$\theta =(\phi_{1},\ldots,\phi_{k},\rho_{1},\ldots,\rho_{k},\pi_{1},\ldots,\pi_{k-1}),$$ and the model (\ref{mod.mix}) in reparameterized form then equals:
\begin{eqnarray} \label{mod-mix2}
& &Y_i|X_i\ \mbox{independent for}\ i=1,\ldots ,n,\nonumber\\
& &Y_i|X_i = x \sim h_{\theta}(y|x)dy\ \mbox{for}\ i=1,\ldots ,n,\nonumber\\
& &h_{\theta}(y|x) = \sum_{r=1}^{k}\pi_{r} \frac{\rho_r}{\sqrt{2 \pi}} 
\exp(-\frac{1}{2}(\rho_r y - x^T \phi_r)^2)\\
& &\theta =
(\phi_{1},\ldots,\phi_{k},\rho_{1},\ldots,\rho_{k},\pi_{1},\ldots,\pi_{k-1})\in
\R^{kp}\times\R_{>0}^{k}\times\Pi\nonumber\\
& &\Pi = \{\pi; \pi_r > 0\ \mbox{for}\ r=1,\ldots ,k-1\ \mbox{and}\
\sum_{r=1}^{k-1} \pi_r < 1\}\nonumber,
\end{eqnarray} 
with $\pi_k = 1 - \sum_{r=1}^{k-1} \pi_r$.
This is the main model we are analyzing and working with. 

The log-likelihood function of this model equals
\begin{eqnarray}\label{loglik-mix}
\ell(\theta;Y)=\sum_{i=1}^{n}\log\left(\sum_{r=1}^{k}\pi_{r}\frac{\rho_{r}}{\sqrt{2\pi}}
  \exp(-\frac{1}{2} (\rho_{r}Y_{i}-X_{i}^T\phi_{r})^{2})\right).
\end{eqnarray}
Since we want to deal with the $p \gg n$ case, we have to
regularize the maximum likelihood estimator (MLE) in order to obtain
reasonably accurate estimates. We propose below some $\ell_1$-norm penalized MLE
which is different from a naive $\ell_1$-norm penalty for the MLE in the
non-transformed model (\ref{mod.mix}). Furthermore, it is well known that
the (log-) likelihood function is generally unbounded. We will see in
Section \ref{subsec.MLEmixpen} that our penalization will mitigate this
problem.  

\section{$\ell_1$-norm penalized maximum likelihood
  estimator} \label{sec:penregr}  

We argue first for the case of a (non-mixture) linear model why the
reparameterization above in Section \ref{subsec.reparmix} is useful and
quite natural.  

\subsection{$\ell_1$-norm penalization for reparameterized linear
  models}\label{subsec.reparlin}

Consider a Gaussian linear model
\begin{eqnarray}\label{mod.lin}
& &Y_i = \sum_{j=1}^p \beta_j X_i^{(j)} + \varepsilon_i, \ i=1,\ldots ,n,\nonumber\\
& &\varepsilon_1,\ldots ,\varepsilon_n\ \mbox{i.i.d.}\ \sim {\cal N}(0,\sigma^2),
\end{eqnarray}
where $X_i$ are either fixed or random covariates. 
In short, we often write
\begin{eqnarray*}
Y = \bx \beta + \eps,
\end{eqnarray*}
with $n \times 1$ vectors $Y$ and $\varepsilon$, $p \times 1$ vector
$\beta$ and $n \times p$ matrix $\bx$. In the sequel, $\|\cdot\|$ denotes
the Euclidean norm. The 
$\ell_1$-norm penalized estimator, called the Lasso
(\citet{tibshirani96regression}), is defined as
\begin{eqnarray}\label{lasso}
\hat{\beta}_{\lambda} = \argmin_{\beta} n^{-1} \|Y - \bx \beta\|^2 + \lambda
\sum_{j=1}^p |\beta_j|. 
\end{eqnarray}
Here $\lambda$ is a non-negative regularization parameter.
The Gaussian assumption is not crucial in model (\ref{mod.lin}) but it is
useful to make connections to the likelihood framework. The Lasso estimator
in (\ref{lasso}) is equivalent to minimizing the penalized negative
log-likelihood \linebreak$n^{-1} \ell(\beta;Y_1,\ldots ,Y_n)$ as a function of the
regression coefficients 
$\beta$ and using the
$\ell_1$-penalty $\|\beta\|_1 = \sum_{j=1}^p |\beta_j|$: equivalence here means that we obtain the same estimator for a  potentially different tuning
parameter. But the Lasso estimator in (\ref{lasso}) does not provide an
estimate of the nuisance
parameter $\sigma^2$.

In mixture models, it will be crucial to have a good estimator of
$\sigma^2$ and the role of the scaling of the variance parameter is much more
important than in homogeneous regression models. Hence, it is important to
take $\sigma^2$ into the definition and 
optimization of the penalized maximum likelihood estimator: we could
proceed with the following estimator,
\begin{align}\label{lasso2}
\hat{\beta}_{\lambda},\hat{\sigma}_{\lambda}^2 &= \argmin_{\beta,\sigma^2}-
n^{-1} \ell(\beta,\sigma^2;Y_1,\ldots ,Y_n) + \lambda \|\beta\|_1 \nonumber \\
&= \argmin_{\beta,\sigma^2} \log(\sigma)+\|Y - \bx \beta\|^2/(2 n
\sigma^2) + \lambda \|\beta\|_1.
\end{align}
Note that we are penalizing only the $\beta$-parameter. However, the scale
parameter estimate $\hat{\sigma}^2_{\lambda}$ is influenced indirectly by
the amount of shrinkage $\lambda$. 

There are two main drawbacks of the estimator in
(\ref{lasso2}). First, it 
is not equivariant \citep{lehmann} under scaling of the response. More
precisely, consider the
transformation
\begin{eqnarray}\label{eq:trans}
Y'=b Y,\quad \beta'=b\beta,\quad \sigma'=b\sigma \quad(b > 0)
\end{eqnarray} 
which leaves model (\ref{mod.lin}) invariant. A reasonable estimator based
on transformed data $Y'$ should lead to estimators $\hat{\beta'}, \hat{\sigma}'$ which
are related to $\hat{\beta}, \hat{\sigma}$ through
$\hat{\beta'}=b\hat{\beta}$ and $\hat{\sigma}'=b\hat{\sigma}$. This is
not the case for the estimator in (\ref{lasso2}). Secondly, and as important
as the first issue, the optimization
in (\ref{lasso2}) is non-convex and hence, some of the major computational
advantages of Lasso for high-dimensional problems is lost. We address these
drawbacks by using the penalty term $\lambda \frac{\|\beta\|_1}{\sigma}$
leading to the following estimator:
\begin{eqnarray*}
\hat{\beta}_{\lambda},\hat{\sigma}_{\lambda}^2 =
\argmin_{\beta,\sigma^2} \log(\sigma) + \|Y - \bx \beta\|^2/(2 n
\sigma^2) + \lambda \frac{\|\beta\|_1}{\sigma}.
\end{eqnarray*}
This estimator is equivariant under the scaling transformation
(\ref{eq:trans}), i.e., the estimators $\hat{\beta'},\hat{\sigma}'$ based on
$Y'$ transform as
$\hat{\beta'}=b\hat{\beta}$ and $\hat{\sigma}'=b\hat{\sigma}$. Furthermore,
it penalizes the $\ell_{1}$-norm of the coefficients and small variances
$\sigma^2$ simultaneously which has some close relations to the Bayesian Lasso
\citep{bayesianlasso}. For the latter, a Bayesian approach is used with a
conditional Laplace prior specification of the form 
\[
 p(\beta|\sigma^{2})=\prod_{j=1}^{p}\frac{\lambda}{2\sqrt{\sigma^{2}}} \exp(-\lambda\frac{|\beta_{j}|}{\sqrt{\sigma^{2}}})
\] 
and a noninformative scale-invariant marginal prior $p(\sigma^{2}) =
1/\sigma^{2}$ for $\sigma^{2}$. \cite{bayesianlasso} argue that
conditioning on $\sigma^{2}$ 
is important because it guarantees a unimodal full posterior. 

Most importantly, we can re-parameterize to achieve convexity of the
optimization problem 
\begin{eqnarray*}
\phi_j = \beta_j/\sigma,\ \ \rho = \sigma^{-1}.
\end{eqnarray*}
This then yields the following estimator which is equivariant under scaling
and whose computation involves convex optimization:
\begin{eqnarray}\label{eq:convexplik}
\hat{\phi}_{\lambda},\hat{\rho}_{\lambda} =
\argmin_{\phi,\rho}-\log(\rho)+\frac{1}{2 n}\|\rho Y-\bx
\phi\|^{2}+\lambda\|\phi\|_{1}. 
\end{eqnarray}

From an algorithmic point of view, fast algorithms are available to solve
the optimization in (\ref{eq:convexplik}). Shooting algorithms \citep{Fu98}
with coordinate-wise descent 
are especially suitable, as demonstrated by, e.g., \cite{friedman07fastlasso}, \cite{meier06group} or
\cite{friedmanetal08}. We describe in Section \ref{subsec:emalgmix} an
algorithm for estimation in a mixture of regressions model, a more complex
task than 
the optimization for (\ref{eq:convexplik}). As we will see in Section
\ref{subsec:emalgmix}, we 
will make use of the Karush-Kuhn-Tucker (KKT) conditions in the M-Step of a
generalized EM algorithm. For the simpler criterion in
(\ref{eq:convexplik}) for a non-mixture model, the KKT conditions imply the
following which we state without a proof. Denote by $\langle \cdot, \cdot
\rangle$ the inner product in $n$-dimensional Euclidean space and by
$\bx_j$ the $j$th column vector of $\bx$. 
 \\ 
\begin{proposition}\label{prop:kkt}
Every solution $(\hat{\phi},\hat{\rho})$ of (\ref{eq:convexplik}) satisfies:
\begin{eqnarray*}
\begin{array}{ccc}
-\hat{\rho} \langle \bx_{j},Y\rangle+\langle \bx_{j},\bx\hat{\phi}\rangle+n\lambda\, \mathrm{sgn}(\hat{\phi}_{j})&=&0\qquad\textrm{if}\qquad\hat{\phi}_{j}\neq 0,\\
|-\hat{\rho} \langle \bx_{j},Y\rangle+\langle \bx_{j},\bx\hat{\phi}\rangle|&\leq&n\lambda\qquad\textrm{if}\qquad\hat{\phi}_{j}=0,
\end{array}
\end{eqnarray*}
and
\begin{eqnarray*}
\hat{\rho}=\frac{\langle Y,\bx\hat{\phi}\rangle+\sqrt{\langle
    Y,\bx \hat{\phi}\rangle^{2}+4\|Y\|^{2}n}}{2\|Y\|^{2}}. 
\end{eqnarray*}
\end{proposition} 

\subsection{$\ell_1$-norm penalized MLE for mixture of Gaussian
  regressions}\label{subsec.MLEmixpen} 
Consider the mixture of Gaussian regressions model in
(\ref{mod-mix2}). Assuming that $p$ is large, we want to regularize the
MLE. In the spirit of the approach in (\ref{eq:convexplik}), we propose
for the unknown parameter $\theta =
(\phi_{1},\ldots,\phi_{k},\rho_{1},\ldots,\rho_{k},\pi_{1},\ldots,\pi_{k-1})$
the estimator:
\begin{align}
\hat{\theta}_{\lambda}^{(\gamma)} &= \argmin_{\theta \in \Theta}
- n^{-1} \ell_{\mathrm{pen},\lambda}^{(\gamma)}(\theta),\label{eq:plik0}\\                
- n^{-1}\ell_{\mathrm{pen},\lambda}^{(\gamma)}(\theta)&=-n^{-1}\sum_{i=1}^{n}\log\left(\sum_{r=1}^{k}\pi_{r}\frac{\rho_{r}}{\sqrt{2\pi}}\exp(-\frac{1}{2}(\rho_{r}Y_{i}-X_{i}^T\phi_{r})^{2})\right)\nonumber\\
&+\lambda\sum_{r=1}^{k}\pi_{r}^{\gamma}\|\phi_{r}\|_{1},\label{eq:plik1}\\ 
\Theta &= \R^{kp} \times \R_{>0}^{k}\times \Pi,\label{eq:parrange}
\end{align}
where $\Pi = \{\pi; \pi_r > 0\ \mbox{for}\ r=1,\ldots ,k-1\ \mbox{and}\
\sum_{r=1}^{k-1} \pi_r < 1\}$ with $\pi_k = 1 - \sum_{r=1}^{k-1}
\pi_r$. The value of $\gamma \in \{0, 1/2,1\}$ 
parameterizes three different penalties.  

The first penalty function with $\gamma = 0$ is independent of the
component probabilities $\pi_{r}$. As we will see in Sections
\ref{subsec:emalgmix} and \ref{subsec.numconv}, the
optimization for computing $\hat{\theta}_{\lambda}^{(0)}$ is
easiest, and we establish a rigorous result about numerical
convergence of a generalized EM algorithm. The penalty with $\gamma = 0$
works fine if the components are not very unbalanced, i.e., the true
$\pi_r$'s aren't too different. In case of strongly unbalanced components,
the penalties with values $\gamma \in \{1/2,1\}$ are to be preferred,
at the price of having to pursue a more difficult optimization problem. 
The value
of $\gamma = 1$ has been proposed by \cite{khalili} for the naively
parameterized likelihood from model (\ref{mod.mix}). We will report in
Section \ref{subsec.simul} about empirical comparisons with the three
different penalties 
involving $\gamma \in \{0,1/2,1\}$.  

All three penalty functions involve the $\ell_{1}$-norm of the component
specific 
ratios $\phi_r = \frac{\beta_{r}}{\sigma_{r}}$ and hence small variances are
penalized. The penalized criteria therefore stay finite whenever
$\sigma_{r}\rightarrow 0$: this is in sharp contrast to the unpenalized MLE
where the likelihood is unbounded if $\sigma_r \to 0$; see, for example,
\cite{McLachlan}.\\
\begin{proposition}\label{prop:bounded}
Assume that $Y_i\neq0$ for all $i=1,\ldots ,n$. Then the penalized negative
log-likelihood $-n^{-1} \ell_{\mathrm{pen},\lambda}^{(0)}(\theta)$ is bounded from 
below for all values $\theta \in \Theta$ from (\ref{eq:parrange}).
\end{proposition}
A proof is given in Appendix \ref{app2}. Even though
Proposition~\ref{prop:bounded} is only stated and proved for the penalized 
negative log-likelihood with $\gamma=0$, we expect that the statement is
also true for $\gamma=1/2$ or $1$.

Due to the $\ell_1$-norm penalty, the estimator is shrinking some of the
components of $\phi_1,\ldots ,\phi_k$ exactly to zero, depending on the
magnitude of the regularization parameter $\lambda$. Thus, we can do
variable selection as follows. Denote by 
\begin{eqnarray}\label{eq:varsel} 
\widehat{S} =\left\{(r,j);\ 1 \le r \le k,\ 1 \le j \le p, \hat{\phi}_{r,j}
  \neq 0 \right\}.
\end{eqnarray} 
Here, $\hat{\phi}_{r,j}$ is the $j$th coefficient of the estimated regression parameter
$\hat{\phi}_{r}$ belonging to mixture component $r$. The set $\widehat{S}$
denotes the collection of non-zero estimated, i.e., selected, regression
coefficients in the $k$ mixture 
components. Note that no significance testing is involved, but, of course,
$\widehat{S} = \widehat{S}_{\lambda}^{(\gamma)}$
depends on the specification of the regularization parameter $\lambda$ and
the type of penalty described by $\gamma$. 

\subsection{Adaptive $\ell_1$-norm penalization}\label{subsec.adaptive}
A two-stage adaptive $\ell_1$-norm penalization for linear models has been
proposed by \cite{zou05adaptive}, called the adaptive Lasso. It is an effective way to address
some bias problems of the (one-stage) Lasso which may employ strong
shrinkage of coefficients corresponding to important variables.  

The two-stage adaptive $\ell_1$-norm penalized estimator for a mixture of
Gaussian regressions is defined as follows. Consider an initial estimate
$\theta^{\mathrm{ini}}$, for example, from the estimator in~(\ref{eq:plik0}). The
adaptive criterion to be minimized involves a re-weighted $\ell_1$-norm
penalty term:
\begin{align}\label{eq:plikadapt}
-n^{-1} \ell_{\mathrm{adapt}}^{(\gamma)}(\theta)&=
-n^{-1}\sum_{i=1}^{n}\log\left(\sum_{r=1}^{k}\pi_{r}\frac{\rho_{r}}{\sqrt{2\pi}}\exp(-\frac{1}{2}(\rho_{r}Y_{i}-X_{i}^T\phi_{r})^{2})\right)\nonumber\\
&+\lambda\sum_{r=1}^{k}\pi_{r}^{\gamma}\sum_{j=1}^{p}w_{r,j}|\phi_{r,j}|,\\
w_{r,j}&=\frac{1}{|\phi_{r,j}^{\mathrm{ini}}|},\ \theta =
(\rho_{1},\ldots,\rho_{k},\phi_{1},\ldots,\phi_{k},\pi_{1}, \ldots,\pi_{k-1}),\nonumber
\end{align}
where $\gamma \in \{0,1/2,1\}$. The estimator is then defined as 
\begin{eqnarray}\label{eq.adaptest} 
\hat{\theta}_{\mathrm{adapt};\lambda}^{(\gamma)} = \argmin_{\theta \in \Theta}
-n^{-1} \ell_{\mathrm{adapt}}^{(\gamma)}(\theta),
\end{eqnarray}
where $\Theta$ is as in (\ref{eq:parrange}).

The adaptive Lasso in linear models has better variable selection properties
than the Lasso, see \citet{zou05adaptive}, \citet{HMZ08},
\citet{zhougeerpb09}. We present some theory for the adaptive estimator in
FMR models in Section \ref{sec.asympt}. Furthermore, we report some
empirical results in Section \ref{subsec.simul} indicating that the
two-stage adaptive method often outperforms the one-stage
$\ell_1$-penalized estimator. 
 
\subsection{Selection of the tuning parameters}\label{sec:select.tuning}
The regularization parameters to be selected are the number of components
$k$, the penalty parameter $\lambda$ and we may also want to select the
type of the penalty function, i.e., selection of $\gamma$. 

One possibility is to use a modified BIC criterion which minimizes
\begin{eqnarray}\label{eq:bic}
\mathrm{BIC}&=&-2\ell(\hat{\theta}_{\lambda,k}^{(\gamma)})+\log(n)\mathrm{d_{e}},
\end{eqnarray}
over a grid of candidate values for $k$, $\lambda$ and maybe also
$\gamma$. Here, $\hat{\theta}_{\lambda,k}^{(\gamma)}$ denotes the estimator in
(\ref{eq:plik0}) using the parameters $\lambda, k, \gamma$ in
(\ref{eq:plik1}), and $-\ell(\cdot)$ is the negative
log-likelihood. Furthermore, $\mathrm{d_{e}} = 
k+(k-1) + \sum_{j=1,\ldots,p; r=1,\ldots,k}\mathop{1_{\{\hat{\phi}_{r,j} \neq
    0\}}}$ is the effective number of parameters \citep{pan}.  

Alternatively, we may use a cross-validation scheme for tuning parameter
selection minimizing some cross-validated negative log-likelihood. 

Regarding the grid of candidate values for $\lambda$, we consider 
$0\leq\lambda_{1}<\cdots<\lambda_{M}\leq\lambda_{\mathrm{max}}$, where $\lambda_{\mathrm{max}}$ is
given by
\begin{equation}\label{eq:lambdamax}
\lambda_{\mathrm{max}}=\max_{j=1,\ldots,p}\left|\frac{\langle
    Y,\bx_{j}\rangle}{\sqrt{n}\|Y\|}\right|. 
\end{equation}
At $\lambda_{\mathrm{max}}$, all coefficients $\hat{\phi}_{j},\ (j=1,\ldots,p)$ of 
the one-component model are exactly zero. Equation (\ref{eq:lambdamax})
easily follows from Proposition \ref{prop:kkt}.

For the adaptive $\ell_1$-norm penalized estimator minimizing the criterion
in (\ref{eq:plikadapt}), we proceed analogously but replacing
$\hat{\theta}_{\lambda,k}^{(\gamma)}$ in (\ref{eq:bic}) by
$\hat{\theta}_{\mathrm{adapt};\lambda}^{(\gamma)}$ in (\ref{eq.adaptest}). As
initial estimator in the 
adaptive criterion, we propose to use the estimate in (\ref{eq:plik0})
which is optimally tuned using the modified BIC or some cross-validation
scheme. 

\section{Asymptotics for  fixed $p$ and $k$}\label{sec.asympt}

Following the penalized likelihood theory of \cite{fanli}, we establish
first some asymptotic properties of the estimator in (\ref{eq:plik1}). As in
\cite{fanli}, we assume in this section that the design is random and that the number of covariates $p$ and the number of
mixture components $k$ are fixed as sample size $n \to \infty$. Of course,
this does not reflect a truly high-dimensional scenario, but the theory and
methodology is much easier for this case. An extended theory for $p$
potentially very large in relation to $n$ is presented in Section
\ref{subsec.ashighdim}. 

Denote by $\theta_{0}$ the true parameter. \\
\begin{theorem}\label{theorem:rootn} (Consistency)
Consider model (\ref{mod-mix2}) with random design, fixed $p$ and $k$. If 
$\lambda=O(n^{-1/2})\ (n \to \infty)$ then, under the regularity conditions
(A)-(C) from \cite{fanli} on the joint density function of $(Y,X)$, there exists a local minimizer
$\hat{\theta}_{\lambda}^{(\gamma)}$ of $-n^{-1} \ell^{(\gamma)}_{\mathrm{pen},\lambda}(\theta)$ in
(\ref{eq:plik1}) ($\gamma \in \{0,1/2,1\}$) such that  
\[\sqrt{n}\left(\hat{\theta}_{\lambda}^{(\gamma)}-\theta_{0}\right)=O_P(1).\] 
\end{theorem}

A proof is given in Appendix \ref{app0}. Theorem \ref{theorem:rootn} can be
easily 
misunderstood. It does not guarantee the existence of an asymptotically
consistent sequence of estimates. The only claim is that a clairvoyant
statistician (with pre-knowledge of $\theta_{0}$) can choose a consistent
sequence of roots in the neighborhood of $\theta_0$ \citep{Vandervaart}. In
the case where $-n^{-1} 
\ell^{(\gamma)}_{\mathrm{pen},\lambda}(\theta)$ has a unique minimizer, which is the case for
a FMR model with one component, the resulting
estimator would be root-$n$ consistent. But for a general FMR model with more
than one component and typically several local minimizers, this does not
hold anymore. In this sense, the preceding theorem
might look better than it is.

Next, we present an asymptotic oracle result in the spirit of \cite{fanli}
for the two-stage adaptive procedure described in
Section~\ref{subsec.adaptive}. 
Denote by $S$ the population analogue of (\ref{eq:varsel}),
i.e., the set of non-zero regression coefficients. Furthermore, let
$\theta_{S}= (\{\phi_{r,j};\ (r,j) \in S\},\rho_1,\ldots
,\rho_k,\pi_1,\ldots ,\pi_{k-1})$ be the sub-vector of parameters corresponding to the true non-zero regression
coefficients (denoted by $S$) and analogously for $\hat{\theta}_{S}$. 
\\
\begin{theorem}\label{theorem:oracle} (Asymptotic oracle result for
  adaptive procedure)\\
Consider model (\ref{mod-mix2}) with random design, fixed $p$ and $k$.
If $\lambda=o(n^{-1/2})$, $n \lambda\rightarrow\infty$ and if $\theta^{\mathrm{ini}}$
satisfies $\theta^{\mathrm{ini}} - \theta_0 = O_P(n^{-1/2})$, then, under the regularity
conditions (A)-(C) from \cite{fanli} on the joint density function of $(Y,X)$, there exists a
local minimizer $\hat{\theta}_{\mathrm{adapt};\lambda}^{(\gamma)}$ of
$- n^{-1} \ell_{\mathrm{adapt}}^{(\gamma)}(\theta)$ in (\ref{eq:plikadapt}) ($\gamma
\in \{0,1/2,1\}$) which satisfies:
\begin{enumerate}
\item[1.]\emph{Consistency in variable selection:}
  $\Prob[\widehat{S}_{\mathrm{adapt};\lambda}^{(\gamma)} = S]\rightarrow 1\ (n \to
  \infty)$.  
\item[2.]\emph{Oracle Property:}
  $\sqrt{n}\left(\hat{\theta}_{\mathrm{adapt};\lambda,S}^{(\gamma)}-\theta_{0,S}\right)
  \leadsto^{d}\calN(0,I_{S}(\theta_{0})^{-1})$, where
  $I_{S}(\theta_{0})$ is the Fisher information knowing that 
  $\theta_{S^{c}}=0$ (i.e., the submatrix of the Fisher information at
  $\theta_{0}$ corresponding to the variables in $S$). 
\end{enumerate}
\end{theorem}
A proof is given in Appendix \ref{app0}. As in Theorem \ref{theorem:rootn}, the
assertion of the theorem is only making a statement about \emph{some} local
optimum. Furthermore, this result only holds for the adaptive
criterion with weights $w_{r,j}=\frac{1}{|\phi_{r,j}^{\mathrm{ini}}|}$ coming from a
root-$n$ consistent initial estimator $\theta^{\mathrm{ini}}$: this is a rather strong
assumption given the fact that Theorem \ref{theorem:rootn} only ensures
existence of such an estimator. The non-adaptive estimator with the  
$\ell_{1}$-norm penalty as in (\ref{eq:plik1}) cannot achieve sparsity and
maintain root-$n$ consistency due to the bias problem mentioned in
Section~\ref{subsec.adaptive} (see also \cite{khalili}).

\section[General theory for $p\gg n$]{General theory for high-dimensional setting with non-convex smooth loss}\label{subsec.ashighdim}  

We present here some theory, entirely different from Theorems
\ref{theorem:rootn} and \ref{theorem:oracle}, which reflects some
consistency and optimality
behavior of the $\ell_1$-norm penalized maximum likelihood
estimator for the potentially high-dimensional framework with $p \gg
n$. In particular, we derive some oracle inequality which is non-asymptotic. 
We intentionally present this theory for
$\ell_1$-penalized smooth likelihood problems which are generally
non-convex; $\ell_1$-penalized likelihood estimation in FMR models is then a
special case discussed in Section \ref{subsec.fmrmod-orac}. The following
Sections \ref{subsec.setting} - \ref{subsec.oracle} introduce some
mathematical conditions and derive auxiliary results and a general oracle
inequality (Theorem \ref{th.oracle}); the interpretation of these conditions and
of the oracle result is discussed for the case of FMR models at the end of
Section \ref{subsec.fmrorac}.

\subsection{The setting and notation}\label{subsec.setting}
Let $\{f_{\psi};\ \psi \in \Psi\}$ be a collection of densities
with respect to the Lebesgue measure $\mu$ on $\R$ (i.e., the range for the
response variable). The parameter space
$\Psi$ is assumed to be a bounded subset of some finite-dimensional
space, say
$$\Psi \subset \{ \psi \in \R^d;\ \| \psi \|_{\infty} \le K \},$$
where we have equipped (quite arbitrarily) the space $\R^d$ with the sup-norm
$\| \psi \|_{\infty} = \max_{1 \le j \le d } | \psi_j |$.  
In our setup, the dimension $d$ will be regarded as a fixed constant (which
still covers high-dimensionality of the covariates, as we will see). Then,
equivalent metrics are, e.g., the ones induced by the  
$\ell_q$-norm $\| \psi \|_q =(  \sum_{j=1}^d | \psi_j |_q)^{1/q} $
($q \ge 1$). 

We observe a covariate $X$ in some space ${\cal X} \subset \R^{p}$ and a response variable $Y\in \R $. The true conditional density
 of $Y$ given $X=x$  is assumed to be equal to
$$ f_{\psi_0} (\cdot \vert x ) = f_{\psi_0 (x)} , $$
where
$$\psi_0 (x) \in \Psi , \ \forall \  x \in {\cal X} . $$
That is, we assume that the true conditional density of $Y$ given $x$
is depending on $x$ only 
through some parameter function $\psi_0(x)$. Of course, the introduced
notation also applies to fixed instead of random covariates.

The parameter $\{ \psi_0 (x);\ x \in {\cal X} \} $ is assumed to have a 
nonparametric part of interest $\{ g_0 (x);\ x \in {\cal X} \} $ and a low-dimensional
nuisance part $\eta_0$, i.e., 
$$\psi_0 (\cdot )^T = ( g_0 (\cdot )^T , \eta_0^T ), $$
with
$$ g_0 (x) \in \R^k , \ \forall \ x \in {\cal X} , \ \eta_0 \in \R^m , \
k+m = d. $$ 
In case of FMR models, $g(x)^T = (\phi_1^Tx,\phi_2^Tx,\ldots ,\phi_k^Tx)$
and $\eta$ involves the parameters $\rho_1,\ldots ,\rho_k, \pi_1,\ldots
,\pi_{k-1}$. More details are given in Section \ref{subsec.fmrmod-orac}. 

With minus the log-likelihood as loss function, the so-called excess risk 
$${\cal E} ( \psi \vert \psi_0) = -
\int \log \biggl [ {f_{\psi } \over f_{\psi_0} } \biggr ] f_{\psi_0}
d \mu $$ 
is the Kullback-Leibler information. For fixed covariates $x_1,\ldots
,x_n$, we define the average excess risk
$$\bar {\cal E} ( \psi \vert \psi_0 ) =
{1 \over n } \sum_{i=1}^n {\cal E} \biggl ( \psi (x_i) \biggl \vert
\psi_0 (x_i) \biggr ) , $$
and for random design, we take the expectation $\EE[{\cal E} ( \psi(X)
\vert \psi_0(X))]$. 

\subsubsection{The margin}\label{subsec.margin}

Following \cite{tsybakov04} and \cite{geer06high} we call the behavior of the
excess risk ${\cal E} ( \psi \vert \psi_0)$ near $\psi_0$ the margin. We
will show in Lemma \ref{marginlemma} that the margin is quadratic. 

Denote by 
$$l_{\psi} = \log f_{\psi} $$
the log-density.
Assuming the derivatives exist, we define the score function
$$s_{\psi} = {\partial l_{\psi}  \over \partial \psi} , $$
and the Fisher information
$$I (\psi) = \int s_{\psi} s_{\psi}^T f_{\psi} d \mu = - 
\int {\partial^2 l_{\psi}  \over \partial \psi
\partial \psi^T } f_{\psi} d \mu . $$
Of course, we can then also look at $I(\psi(x))$ using the parameter
function $\psi(x)$. 

In the sequel, we introduce some conditions (Conditions 1 - 5). Their
interpretation for the case of FMR models is given at the end of Section
\ref{subsec.fmrorac}. First, we will assume boundedness of third
derivatives. \\ \\
{\bf Condition 1} {\it It holds that
$$\sup_{\psi\in \Psi} \max_{(j_1 , j_2 , j_3)  \in \{ 1 , \ldots , d \}^3 }
\biggl | 
{\partial^3 \over \partial \psi_{j_1} \partial
\psi_{j_2} \partial \psi_{j_3} } l_{\psi} ( \cdot )\biggr  | \le
G_3 ( \cdot) , $$
where
$$ \sup_{x\in{\cal X}} \int G_3 (y) f_{\psi_0} (y \vert x) d \mu (y) \le
C_3 < \infty . $$}
\\ \\
For a symmetric, positive semi-definite matrix $A$, we let
$\Lambda_{\rm min}^2 (A)$ be its smallest eigenvalue.
\\ \\
{\bf Condition 2} {\it For all $x\in{\cal X}$, 
the Fisher information matrix $I(\psi_0(x))$ is positive definite and, in fact,
$$\Lambda_{\rm min} = 
\inf_{x\in{\cal X}} \Lambda_{\rm min} (I(\psi_0(x))) >0 . $$}
\\ \\
Further we will need the following identifiability condition. 
\\ \\
{\bf Condition 3} {\it For all $\eps > 0$, there exists an
$\alpha_{\eps} > 0 $, such that 
$$ \inf_{x\in{\cal X}} \inf_{\psi \in \Psi \atop \ \  \| \psi - \psi_0 (x) \|_2 > \eps } 
{\cal E} ( \psi \vert \psi_0 (x)) \ge \alpha_{\eps} . $$}
Based on these three conditions, we have the following result:\\
\begin{lemma} \label{marginlemma}
Assume Conditions 1, 2, and 3. Then
$$\inf_{x\in{\cal X}} {{\cal E} ( \psi \vert \psi_0(x)) \over
\| \psi - \psi_0 (x) \|_2^2 } \ge {1 \over c_0^2 } , $$
where
$$c_0^2 = \max \biggl [ {1 \over \eps_0 } , { d K^2 \over \alpha_{\eps_0}}
\biggr ] , \  \eps_0= {3 \Lambda_{\rm min}^2 \over2 d^{3/2} } .$$
\end{lemma}
A proof is given in Appendix \ref{app1}. 

\subsubsection{The empirical process}
We now specialize to the case where 
$$\psi(x)^T= ( g(x)^T , \eta^T ), $$
where (with some abuse of notation)
\begin{eqnarray*}
& &g(x)^T = g_{\phi } (x)^T = ( g_1 (x) , \ldots , g_k (x)),\\
& &g_r (x) = g_{\phi_r } (x) = x^T \phi_r , \ x \in \R^p, \ \phi_r \in
\R^p,\ r=1,\ldots ,k.
\end{eqnarray*}
We also write
\[
\psi_{\vartheta}(x)^{T}=(g_{\phi}(x)^{T},\eta^{T}),\qquad\vartheta^{T}=(\phi_{1}^T,\ldots,\phi_{k}^T,\eta^T)
\]
to make the dependence of the parameter function $\psi(x)$ on $\vartheta$ more explicit.

We will assume that
$$\sup_{x\in{\cal X}} \| \phi^T x \|_{\infty}= \sup_{x\in{\cal X}} \max_{1 \le r \le k} |
\phi_r^T x | \le K. $$ This can be viewed as a combined condition on $\cal
X$ and $\phi$. For example, if $\cal X$ is bounded by a fixed constant this
supremum (for fixed $\phi$) is finite.

Our parameter space is now
\begin{eqnarray}\label{add2}
\tilde \Theta \subset 
\{ \vartheta^T = ( \phi_1^T , \ldots , \phi_k^T , \eta^T );\ \sup_{x\in{\cal X}} \| \phi^T x \|_{\infty}\le K, \ \| \eta \|_{\infty} \le K \}.
\end{eqnarray}
Note that $\tilde \Theta$ is, in principle, $(pk+m)$-dimensional. The true
parameter $\vartheta_0$ is assumed to be an element of $\tilde \Theta$. 

Let us define
$$L_{\vartheta} (x, \cdot) =  \log f_{\psi(x) } (\cdot ),\ \ \psi (x)^T =
\psi_{\vartheta}(x)^T = (g_{\phi} (x)^T , \eta^T ),  
$$
$ \vartheta^T = (\phi_1^T , \ldots , \phi_k^T , \eta^T ),$ and the empirical process for fixed covariates $x_1,\ldots, x_n$ 
$$V_n (\vartheta) = {1 \over n}
\sum_{i=1}^n \left ( L_{\vartheta} (x_i , Y_i) -
\EE \Bigl [ L_{\vartheta} (x_i , Y) \Bigl \vert X = x_i \Bigr ] \right ) .$$
We now fix some $T \ge 1$ and $\lambda_0 \ge 0$ and define the event
\begin{eqnarray}\label{setT}
{\cal T} = \left \{ \sup_{\vartheta^T = ( \phi^T , \eta^T )  \in \tilde \Theta } { 
 \biggl |  V_n ( \vartheta ) - V_n ( \vartheta_0) 
\biggr | \over (
\| \phi - \phi_0 \|_1 + \| \eta - \eta_0 \|_2 ) \vee \lambda_0 } \le T
\lambda_0  \right \}.
\end{eqnarray}


\subsection{Oracle inequality for the Lasso for non-convex loss
  functions}\label{subsec.oracle} 
For an optimality result, we need some condition on the design. 
Denote the active set, i.e., the set of non-zero coefficients, by 
$$S= \{ (r,j);\ \phi_{r,j} \not= 0 \} , \ s =|S| , $$
and let 
$$\phi_J = \{\phi_{(r,j)};\ (r,j) \in J \} , \ J \subseteq \{ 1 , \ldots , k
\}\times\{ 1 , \ldots , p
\}. $$
Further, let 
$$\Sigma_n = {1 \over n} \sum_{i=1}^n x_i x_i^T . $$
\\ \\
{\bf Condition 4} {\it (Restricted eigenvalue condition). 
There exists a constant $\kappa \ge 1 $ such that, for all $\phi \in \R^{pk}$
satisfying 
$$ \| \phi_{S^c} \|_1 \le 6 \| \phi_S \|_1, $$
it holds that
$$\| \phi_S \|_2^2 \le \kappa^2 \sum_{r=1}^k \phi_r^T \Sigma_n \phi_r . $$}


For $\psi (\cdot)^T = (g( \cdot )^T , \eta^T ) $, we use the notation
$$ \| \psi \|_{Q_n}^2 = {1 \over n} \sum_{i=1}^n \sum_{r=1}^k g_r^2 (x_i) +
\sum_{j=1}^m \eta_j^2 . $$
We also write for $g(\cdot ) = (g_1 (\cdot ) , \ldots , g_k (\cdot ) )^T$,
$$\| g \|_{Q_n}^2 ={1 \over n} \sum_{i=1}^n  \sum_{r=1}^k g_r^2 (x_i) . $$
Thus
$$\| g_{\phi} \|_{Q_n}^2 = \sum_{r=1}^k \phi_r^T \Sigma_n \phi_r ,$$ 
and the bound in the restricted eigenvalue condition then reads
$$\| \phi_S \|_2^2 \le \kappa^2 \| g_{\phi} \|_{Q_n}^2. $$ 
Bounding $\| g_{\phi} \|_{Q_n}^2$ in terms of $\| \phi_S \|_2^2$ can be
done directly using, e.g., the Cauchy-Schwarz inequality. The restricted
eigenvalue condition ensures a bound
in the other direction which itself is needed for an oracle inequality.  
Some references about the restricted eigenvalue condition are provided at the end
of Section \ref{subsec.fmrorac}. 

We employ the Lasso-type estimator
\begin{align}\label{lasso-compact}
\hat \vartheta^T =
( \hat \phi^T , \hat \eta^T) = 
\argmin_{\vartheta^T = ( \phi^T , \eta^T) \in \tilde \Theta } - {1 \over n} \sum_{i=1}^n L_{\vartheta} ( x_i , Y_i) + \lambda \sum_{r=1}^k
\| \phi_r \|_1.
\end{align}
We omit in the sequel the dependence of $\hat{\vartheta}$ on $\lambda$. 
Note that we consider here a global minimizer which may be difficult to compute if the
empirical risk $n^{-1} \sum_{i=1}^n L_{\vartheta}(x_i,Y_i)$ is non-convex
in $\vartheta$.
We then write $\|\phi \|_1 = \sum_{r=1}^k \| \phi_r \|_1 $. 
We let
$$\hat \psi (x)^T  = ( g_{\hat \phi}
(x)^T , \hat \eta^T ) , $$ 
which depends only on the estimate $\hat{\vartheta}$, and we denote by 
$$\psi_0 (x)^T = (g_{\phi_0} (x)^T , \eta_0^T ) . $$


\begin{theorem}\label{th.oracle} (Oracle result for fixed design). Assume
  fixed covariates $x_1,\ldots ,x_n$, Conditions 1-3 and 4, and that 
$\lambda \ge 2 T \lambda_0$ for the estimator in
(\ref{lasso-compact}) with $T$ and $\lambda_0$ as in (\ref{setT}). Then on
${\cal T}$, defined in (\ref{setT}), for the average excess  
risk (average Kullback-Leibler loss), 
$$\bar {\cal E} (\hat \psi \vert \psi_0) +2 ( \lambda -T \lambda_0)
\| \hat \phi_{S^c} \|_1 \le 
8 (\lambda + T \lambda_0 )^2 c_0^2 \kappa^2 s,$$
where $c_0$ and $\kappa$ are defined in Lemma \ref{marginlemma} and Condition 4,
respectively.  
\end{theorem}
A proof is given in Appendix \ref{app1}. 
We will give an interpretation of this result in Section~\ref{subsec.fmrorac}, where we specialize to FMR
models. In the case of FMR models, the probability of the set ${\cal T}$ is
large as shown in detail by Lemma \ref{lemm.setT} below. 

Before specializing to FMR models, we present more general results for
lower bounding the probability of the set ${\cal T}$. We make the following
assumption. 
\\ \\
 {\bf Condition 5} {\it For the score function $s_{\vartheta}( \cdot
   )=s_{\psi_{\vartheta}}( \cdot )$, we have 
 $$\sup_{\vartheta \in \tilde{\Theta}} \| s_{\vartheta}( \cdot )
 \|_{\infty} \le G_1 (\cdot),$$ for some function $G_1(\cdot)$.
}

Condition 5 primarily has notational character. Later, in
Lemma~\ref{cor.set} and particularly in Lemma~\ref{lemm.setT}, the function
$G_1 (\cdot)$ needs to be sufficiently regular to ensure small
corresponding probabilities.
 
%

Define
\begin{eqnarray}\label{add1}
\lambda_0 = M_n \log n\sqrt { \log (p\vee n) \over n }.
\end{eqnarray}
As we will see, we usually choose $M_n \asymp \sqrt{\log(n)}$. 
Let $\PP_{\bf x}$ denote the conditional probability given
$(X_1 , \ldots , X_n )= (x_1 , \ldots , x_n ) = {\bf x}$, and with the
expression ${\rm l} \{\cdot\}$ we denote the indicator function.
\\
\begin{lemma}\label{cor.set} Assume Condition 5. We have for
  constants $c_1$, $c_2$ and $c_3$
  depending on $k$ and $K$,
and for all $T \ge c_1$, 
$$\sup_{\vartheta^T = ( \phi^T , \eta^T)  \in \tilde \Theta } { 
 \biggl | V_n ( \vartheta ) - V_n ( \vartheta_0) 
 \biggl| \over (
\| \phi - \phi_0 \|_1 + \| \eta - \eta_0 \|_2 ) \vee \lambda_0 } \le T \lambda_0  ,$$
with $\PP_{\bf x} $ probability at least
$$1- c_2\exp\biggl  [ - {  T^2 \log^2\!n \log (p\vee n)   \over c_3^2 } \biggr ] 
- \PP_{\bf x} \left ( {1 \over n} \sum_{i=1}^n F(Y_i)
> T \lambda_0^2 /(dK) \right ) . $$
where (for $i=1 , \ldots , n$)
$$F(Y_i )= 
G_1(Y_i) {\rm l} \{ G_1(Y_i) > M_n \} +
\EE \Bigl [ G_1(Y) {\rm l} \{ G_1(Y) > M_n \} \Bigl \vert X=x_i \Bigr ]. $$
Regarding the constants $\lambda_0$ and $K$, see (\ref{add1}) and
(\ref{add2}), respectively. 
\end{lemma}
A proof is given in Appendix \ref{app1}. 

\subsection{FMR models}\label{subsec.fmrmod-orac}
In the finite mixture of regressions model from (\ref{mod-mix2}) with $k$
components, the parameter 
is $\vartheta^T = ( \phi^T , \eta^T)$ $=(\phi_{1}^T,\ldots,\phi_{k}^T, \log
\rho_1 , \ldots , \log \rho_k, \log \pi_1 , \ldots , \log \pi_{k-1} ) $,
where the $\rho_r = \sigma_r^{-1} $ 
are the inverse standard deviations in mixture component $r$ and the $\pi_r$ are
the mixture coefficients. For mathematical convenience and simpler
notation, we consider here the $\log$-transformed
$\rho$ and $\pi$ parameters in order to have lower and upper bounds for
$\rho$ and $\pi$. Obviously, there is a one-to-one
correspondence between $\vartheta$ and $\theta$ from
Section~\ref{subsec.reparmix}.  

Let the parameter space be
\begin{align}\label{thet.tilde}
\tilde{\Theta} \subset \{\vartheta^T;&\sup_{x\in{\cal X}}\| \phi^Tx \|_{\infty} \le K , \| \log \rho \|_{\infty} \le K ,-K \le \log \pi_1 \le 0 , \ldots ,\nonumber\\ 
&-K \le  \log \pi_{k-1} \le 0,
 \sum_{r=1}^{k-1} \pi_r <1 \},
\end{align}
and $\pi_k = 1 - \sum_{r=1}^{k-1} \pi_r$. 

We consider the estimator
\begin{align}\label{add.1}
\hat{\vartheta}_{\lambda}=\argmin_{\vartheta \in \tilde{\Theta}}&- n^{-1}
\sum_{i=1}^{n}\log\left(\sum_{r=1}^{k}\pi_{r}\frac{\rho_{r}}{\sqrt{2\pi}}  
  \exp(-\frac{1}{2}(\rho_{r}Y_{i}-X_{i}^T\phi_{r})^{2})\right)+\lambda\sum_{r=1}^{k} \|\phi_{r}\|_{1}.
\end{align}
This is the estimator from Section \ref{subsec.MLEmixpen} with $\gamma =
0$. We emphasize the boundedness of the parameter space by using the
notation $\tilde{\Theta}$. 
In contrast to Section \ref{sec.asympt}, we focus here on any global
minimizer of the penalized negative log-likelihood which is arguably
difficult to compute. 

In the following, we transform the estimator $\hat{\vartheta}_{\lambda}$
to $\hat{\theta}_{\lambda}$ in the parameterization $\theta$ from
Section~\ref{subsec.reparmix}. Using some abuse of
notation we denote the average excess risk by $\bar {\cal E} (\hat
\theta_{\lambda} \vert \theta_0)$. 

\subsubsection{Oracle result for FMR models}\label{subsec.fmrorac}
We specialize now our results from Section~\ref{subsec.oracle} to FMR
models.
\\
\begin{proposition}\label{prop.FMR} For fixed design FMR models as in
  (\ref{mod-mix2}) with $\tilde{\Theta}$ in (\ref{thet.tilde}), Conditions
  1,2 and 3 are 
  met, for appropriate $C_3$,   
$\Lambda_{\rm min}$ and $\{ \alpha_{\eps} \}$, depending on $k$ and $K$.
Also Condition 5 holds with
$$G_1 (y) = {\rm e}^K |y| + K . $$
\end{proposition}

\begin{proof}
This follows from straightforward calculations.
\end{proof}

In order to show that the probability for the set ${\cal T}$ is large, we
invoke Lemma \ref{cor.set} and the following result. \\
\begin{lemma}\label{lemm.setT}
For fixed design FMR models as in (\ref{mod-mix2}) with $\tilde \Theta$ in
(\ref{thet.tilde}), for some constants $c_4$, $c_5$ and $c_6$,  
depending on $k$, and $K$, and for $M_n = c_4  \sqrt {\log n } $ and $n \ge c_6$, the following holds:
$$
 \PP_{\bf x} \left ( {1 \over n} \sum_{i=1}^n F(Y_i)
> c_5 {\log n \over n} \right ) \le {1 \over n} , $$
where (for $i=1 , \ldots , n$)
$$F(Y_i )= 
G_1(Y_i) {\rm l} \{ G_1(Y_i) > M_n \} +
\EE \Bigl [ G_1(Y) {\rm l} \{ G_1(Y) > M_n \} \Bigl \vert X=x_i
\Bigr ],$$
and $G_1(\cdot)$ is as in Proposition \ref{prop.FMR}.  
\end{lemma} 
A proof is given in Appendix \ref{app1}. 

Hence, the oracle result in Theorem \ref{th.oracle} for our $\ell_1$-norm
penalized estimator in the FMR model holds on a set ${\cal
  T}$, summarized in Theorem \ref{th.oraclefmr}, and this 
set ${\cal T}$ has large probability due to Lemma \ref{cor.set} and
Lemma \ref{lemm.setT} as described in the following corollary. \\
\begin{corollary}\label{cor.setT}
For fixed design FMR models as in (\ref{mod-mix2}) with $\tilde \Theta$ in
(\ref{thet.tilde}), we have for constants $c_2,c_4,c_7,c_8$ depending
on $k$ and $K$,
\begin{eqnarray*}
\PP_{\bf x}[{\cal T}] \ge 1- c_2\exp\biggl  [ - {  T^2 \log^2\!n \log (p\vee n)   \over c_7^2 }
\biggr ] - n^{-1}\ \mbox{for all}\ n \ge c_8,
\end{eqnarray*}
where ${\cal T}$ is defined with $\lambda_0 = M_n \log
n\sqrt{\log(p\vee n)\over n}$ and
$M_n = c_4  \sqrt {\log n }$. 
\end{corollary}

\begin{theorem}\label{th.oraclefmr} (Oracle result for FMR
  models). Consider a fixed design FMR model as in (\ref{mod-mix2}) with
  $\tilde \Theta$ in (\ref{thet.tilde}). Assume Condition 4 (restricted
  eigenvalue condition) and that $\lambda \ge 2 T \lambda_0$ for the
  estimator in  
(\ref{add.1}). Then on ${\cal T}$, which has large probability as stated in
Corollary \ref{cor.setT}, for the average excess   
risk (average Kullback-Leibler loss), 
$$\bar {\cal E} (\hat \theta_{\lambda} \vert \theta_0) +2 ( \lambda -T \lambda_0)
\| \hat \phi_{S^c} \|_1 \le 
8 (\lambda + T \lambda_0 )^2 c_0^2 \kappa^2 s,$$
where $c_0$ and $\kappa$ are defined in Lemma \ref{marginlemma} and
Condition 4, respectively.  
\end{theorem}
The oracle inequality of Theorem \ref{th.oraclefmr} has the following
well-known interpretation. First, we obtain
\begin{eqnarray*}
\bar {\cal E} (\hat{\theta}_{\lambda} \vert \theta_0) \le 8 (\lambda + T
\lambda_0 )^2 c_0^2 \kappa^2 s.
\end{eqnarray*}
That is, the average Kullback-Leibler risk is of the order
$O(s \lambda_0^2) = O(s \log^3n \log(p\vee n)/n)$ (take $\lambda=2 T \lambda_0$, use
definition (\ref{add1}) and the 
assumption on $M_n$ in Lemma \ref{lemm.setT} above) which is up to the factor
$\log^3n\log(p\vee n)$ 
the optimal convergence rate if one would know the $s$ non-zero
coefficients. As a second implication we obtain
\begin{eqnarray*}
\| \hat \phi_{S^c} \|_1 \le 
4 (\lambda + T \lambda_0 ) c_0^2 \kappa^2 s,
\end{eqnarray*}
saying that the noise components in $S^c$ have small estimated values
(e.g., its $\ell_1$-norm converges to zero at rate $O(s \lambda_0)$).

Note that the Conditions 1, 2, 3 and 5 hold automatically for FMR models, as
described in Proposition \ref{prop.FMR}. We do require a restricted
  eigenvalue condition on the design, here Condition 4. In fact, for the Lasso or
  Dantzig selector in linear models, restricted 
eigenvalue conditions \citep{koltch09,bickel07dantzig} are considerably
weaker than coherence conditions \citep{Bunea:07,cai09b} or  
assuming the restricted isometry property \citep{candes2005decoding,cai09};
for an overview among the relations, see \cite{sarapeter09}.  

\subsubsection{High-dimensional consistency of FMR models}
We finally give a consistency result for FMR models under weaker
conditions than the oracle result from Section~\ref{subsec.fmrorac}. Denote by $\theta_0$ the true parameter
vector in a FMR model. In contrast to Section~\ref{sec.asympt}, the number
of covariates $p$ can grow with the 
number of observations $n$. Therefore, also the true parameter $\theta_0$
depends on $n$. To guarantee consistency we have to assume some sparsity
condition, i.e., the $\ell_{1}$-norm of the true parameter can only grow
with $o(\sqrt{n/(\log^3n \log(p\vee n))})$.\\
\begin{theorem}\label{theorem:consist-higd} (Consistency).
Consider a fixed design FMR model (\ref{mod-mix2}) with $\tilde{\Theta}$ in
(\ref{thet.tilde}) and fixed $k$. Moreover, assume that 
$$\|\phi_0\|_1 = \sum_{r=1}^k \|\phi_{0,r}\|_1 = o(\sqrt{n/(\log^3n \log(p\vee n))})\ (n
\to \infty).$$ 
If $\lambda = C \sqrt{\log^3n \log(p\vee n)/n}$ for some $C>0$
sufficiently large, then any (global) minimizer
$\hat{\theta}_{\lambda}$ as in (\ref{add.1})
satisfies
\begin{eqnarray*}
\bar {\cal E} (\hat \theta_{\lambda} \vert \theta_0) =
o_P(1)\ (n \to \infty). 
\end{eqnarray*}
\end{theorem}
A proof is given in Appendix \ref{app1}. 
The (restricted eigenvalue) Condition 4 on 
the design is not required; this is typical for a high-dimensional
consistency result, see \cite{greenshtein03persistency} for the Lasso in
linear models. 

\section{Numerical optimization}\label{sec.optim} 
We present a generalized EM (GEM) algorithm for optimizing the criterion
in (\ref{eq:plik1}) in Section \ref{subsec:emalgmix}. In Section
\ref{subsec.activeset} and \ref{subsec.initialisation}, we give further
details on speeding-up and on initializing the algorithm. Finally, we
discuss numerical convergence properties in Section
\ref{subsec.numconv}. For the convex penalty ($\gamma = 0$) function, we prove 
convergence to a stationary point.  


\subsection{GEM algorithm for optimization}\label{subsec:emalgmix}

Maximization of the log-likelihood of a mixture density is often done using
the traditional EM algorithm of \cite{Dempster}. Consider the complete log-likelihood
\begin{eqnarray*}
\ell_{c}(\theta;Y,\Delta)&=&\sum_{i=1}^{n}\sum_{r=1}^{k}\Delta_{i,r}\log\left(\frac{\rho_{r}}{\sqrt{2\pi}}e^{-\frac{1}{2}(\rho_{r}Y_{i}-X_{i}^{T}\phi_{r})^{2}}\right)+\Delta_{i,r}\log(\pi_{r}).
\end{eqnarray*} 
Here $(\Delta_{i,1},\ldots,\Delta_{i,k})$, $i=1,\ldots,n$, are i.i.d. unobserved multinomial variables showing the
component-membership of the $i$th observation in the FMR model:
$\Delta_{i,r}=1$ if observation $i$ belongs to component $r$, and
$\Delta_{i,r}=0$ otherwise. The expected complete (scaled) negative
log-likelihood is then 
\begin{eqnarray*}
\mathop{Q}(\theta|\theta')&=&- n^{-1} \EE_{\theta'}[\ell_{c}(\theta;Y,\Delta)|Y],
\end{eqnarray*}
and the expected complete (scaled) penalized negative log-likelihood is
\begin{eqnarray*}
\mathop{Q_{\mathrm{pen}}}(\theta|\theta')&=&\mathop{Q}(\theta|\theta')+\lambda\sum_{r=1}^{k}\pi_{r}^{\gamma}\|\phi_{r}\|_{1}.
\end{eqnarray*}
The EM algorithm works by alternating between the E- and M-Step. Denote the
parameter value at EM-iteration $m$ by $\theta^{(m)}$ ($m =
0,1,2,\ldots$), where $\theta^{(0)}$ is a vector of starting
values. \\ \\
\textbf{E-Step}: Compute $\mathop{Q}(\theta|\theta^{(m)})$,
or equivalently for $r=1,\ldots ,k$ and $i=1,\ldots ,n$
\begin{eqnarray*}
\hat{\gamma}_{i,r}&=&\EE_{\theta^{(m)}}[\Delta_{i,r}|Y]=\frac{\pi_{r}^{(m)}\rho_{r}^{(m)}e^{-\frac{1}{2}(\rho_{r}^{(m)}Y_{i}-X_{i}^{T}\phi_{r}^{(m)})^{2}}}{\sum_{l=1}^{k}\pi_{l}^{(m)}\rho_{l}^{(m)}e^{-\frac{1}{2}(\rho_{l}^{(m)}Y_{i}-X_{i}^{T}\phi_{l}^{(m)})^{2}}}.
\end{eqnarray*}
\\ \\
\textbf{Generalized M-Step}: Improve $\mathop{Q_{\mathrm{pen}}}(\theta|\theta^{(m)})$ w.r.t $\theta \in \Theta$.

\begin{itemize}
\item[a)]\emph{Improvement with respect to $\pi=(\pi_1,\ldots,\pi_k)$:}

fix $\phi$ at the present value $\phi^{(m)}$ and improve
\begin{eqnarray}\label{eq:piupdate}
-n^{-1} \sum_{i=1}^{n}\sum_{r=1}^{k}\hat{\gamma}_{i,r}\log(\pi_{r})+\lambda\sum_{r=1}^{k}\pi_{r}^{\gamma}\|\phi^{(m)}_{r}\|_{1} 
\end{eqnarray}
with respect to the probability simplex
\begin{eqnarray*}
\{\pi; \pi_r > 0\ \mbox{for}\ r=1,\ldots ,k\ \mbox{and}\
\sum_{r=1}^{k} \pi_r = 1\}.
\end{eqnarray*}
Denote by $\bar{\pi}^{(m+1)}=\frac{\sum_{i=1}^{n}\hat{\gamma}_{i}}{n}$ which is a
feasible point of the simplex. 
We propose to update $\pi$ as
\begin{eqnarray*}
\pi^{(m+1)} = \pi^{(m)}+t^{(m)}(\bar{\pi}^{(m+1)}-\pi^{(m)})
\end{eqnarray*}
where $t^{(m)} \in (0,1]$. In practice, $t^{(m)}$ is chosen to be the largest
value in the grid $\{\delta^{k};\ k=0,1,2,\ldots \}$ ($0<\delta<1$) such that
(\ref{eq:piupdate}) is not increased. In our examples, $\delta=0.1$ worked
well.

\item[b)]\emph{Coordinate descent improvement with respect to $\phi$ and $\rho$:}

A simple calculation shows that the M-Step decouples for each component into $k$ distinct
optimization problems of the form
\begin{equation}\label{eq:mstep}
-\log(\rho_{r})+\frac{1}{2
  n_{r}}\|\rho_{r}\tilde{Y}-\mathbf{\tilde{X}}\phi_{r}\|^{2}+\frac{n\lambda}{n_{r}}\left(\pi_{r}^{(m+1)}\right)^{\gamma}\|\phi_{r}\|_{1},\ \
r=1,\ldots,k
\end{equation}
with
\begin{eqnarray*}
n_{r} =\sum_{i=1}^{n}\hat{\gamma}_{i,r},\ \
(\tilde{Y_{i}},\tilde{X_{i}})=\sqrt{\hat{\gamma}_{i,r}}(Y_{i},X_{i}),\qquad
r=1,\ldots,k. 
\end{eqnarray*}
Problem (\ref{eq:mstep}) has the same form as (\ref{eq:convexplik}); in
particular, it is convex in $(\rho_{r},\phi_{r,1},\ldots,\phi_{r,p})$. Instead of fully
optimizing (\ref{eq:mstep}), we only minimize with respect to each of the
coordinates, holding the other coordinates at their current value. Closed-form coordinate updates
can easily be computed for each component $r$ ($r=1,\ldots,k$) using
Proposition \ref{prop:kkt}:
\begin{eqnarray*}
\rho_{r}^{(m+1)}&=&\frac{\langle
  \tilde{Y},\mathbf{\tilde{X}}\phi_{r}^{(m)}\rangle+\sqrt{\langle
    \tilde{Y},\mathbf{\tilde{X}}\phi_{r}^{(m)}\rangle^{2}+4\|\tilde{Y}\|^{2}n_{r}}}{2\|\tilde{Y}\|^{2}},
\end{eqnarray*}

\begin{eqnarray*}
\phi_{r,j}^{(m+1)}&=&\left\{\begin{array}{ll}
0&\textrm{if $|S_{j}| \le n\lambda\left(\pi_{r}^{(m+1)}\right)^{\gamma}$},\vspace{0.2cm}\\
\frac{(n\lambda\left(\pi_{r}^{(m+1)}\right)^{\gamma}-S_{j})}{\|\mathbf{\tilde{X}}_{j}\|^{2}}&\textrm{if
  $S_{j}>n\lambda\left(\pi_{r}^{(m+1)}\right)^{\gamma}$},\vspace{0.2cm}\\
-\frac{(n\lambda\left(\pi_{r}^{(m+1)}\right)^{\gamma}+S_{j})}{\|\mathbf{\tilde{X}}_{j}\|^{2}}&\textrm{if
  $S_{j}<-n\lambda\left(\pi_{r}^{(m+1)}\right)^{\gamma}$},\end{array}\right.
\end{eqnarray*}
where $S_{j}$ is defined as
\[S_{j}=-\rho_{r}^{(m+1)}\langle \mathbf{\tilde{X}}_{j},\tilde{Y}\rangle+
\sum_{s<j}\phi_{r,s}^{(m+1)}\langle
\mathbf{\tilde{X}}_{j},\mathbf{\tilde{X}}_{s}\rangle+
\sum_{s>j}\phi_{r,s}^{(m)}\langle
\mathbf{\tilde{X}}_{j},\mathbf{\tilde{X}}_{s}\rangle\]
and $j=1,\ldots,p$.
\end{itemize}
Because we only improve $\mathop{Q_{\mathrm{pen}}}(\theta|\theta^{(m)})$ instead of
a full minimization, see
M-Step a) and b), this is a generalized EM (GEM) algorithm. We call it the block
coordinate descent 
generalized EM algorithm (BCD-GEM); the word block refers to the fact that
we are updating all components of $\pi$ at once. Its numerical properties
are discussed 
in Section \ref{subsec.numconv}. 
\begin{remark}\label{remark:remark}
For the convex penalty function with $\gamma=0$, a minimization with
respect to $\pi$
in M-Step a) is achieved with
$\pi^{(m+1)}=\frac{\sum_{i=1}^{n}\hat{\gamma}_{i}}{n}$, i.e., using
$t^{(m)}=1$. Then, our M-Step corresponds to exact coordinate-wise minimization
of $\mathop{Q_{\mathrm{pen}}}(\theta|\theta^{(m)})$. 
\end{remark}
\subsection{Active set algorithm}\label{subsec.activeset}
There is a simple way to speed-up the algorithm described above.
When updating the coordinates $\phi_{r,j}$ in the M-Step b), we
restrict ourselves during every 10 EM-iterations to the current active set
(the non-zero coordinates) and 
visit the remaining coordinates every 11th EM-iteration to update the
active set. In very high-dimensional and sparse settings, this leads to a
remarkable decrease in computational times. A similar active set strategy is
also used in \cite{friedman07fastlasso} and \cite{meier06group}. We
illustrate in Section~\ref{subsec.comptiming} the gain of speed when
staying during every 10 EM-iterations within the active set. 

\subsection{Initialization}\label{subsec.initialisation}
The algorithm of Section~\ref{subsec:emalgmix} requires the specification of
starting values $\theta^{(0)}$. 
We found empirically that the following initialization works well. For each
observation $i$, $i=1,\ldots,n$, draw randomly a class $\kappa\in 
\{1,\ldots,k\}$. Assign for observation $i$ and the corresponding component $\kappa$ the weight
$\tilde{\gamma}_{i,\kappa}=0.9$ and weights $\tilde{\gamma}_{i,r}=0.1$ for
all other 
components. Finally, normalize $\tilde{\gamma}_{i,r}$, $r=1,\ldots,k$, to
achieve that summing over the indices $k$ yields the value one, to get the
normalized values $\hat{\gamma}_{i,r}$. Note that this can be viewed as an
initialization of the E-Step. In 
the M-Step which follows afterwards, we update all coordinates from the
initial values 
$\phi^{(0)}_{r,j}=0$, $\rho^{(0)}_{r}=2$, $\pi^{(0)}_{r}=1/k$,
$r=1,\ldots,k$, $j=1,\ldots,p$.

\subsection{Numerical convergence of the BCD-GEM
  algorithm}\label{subsec.numconv} 
We address here the convergence properties of the BCD-GEM
algorithm described in Section \ref{subsec:emalgmix}. 
A detailed account of the convergence properties of the EM algorithm in a
general setting has been given by \cite{wu}. Under regularity
conditions 
including differentiability and continuity, convergence to stationary points
is proved for the EM algorithm. 
For the GEM algorithm, similar statements are true under conditions which are
often hard to verify.

As a GEM algorithm, our BCD-GEM algorithm has the descent property which
means that the criterion function is reduced in each iteration:
\begin{equation}
- n^{-1} \ell^{(\gamma)}_{\mathrm{pen},\lambda}(\theta^{(m+1)})\leq -n^{-1}
\ell_{\mathrm{pen},\lambda}^{(\gamma)}(\theta^{(m)}). 
\end{equation}
Since $-n^{-1} \ell_{\mathrm{pen},\lambda}^{(0)}(\theta)$ is bounded from below
(Proposition \ref{prop:bounded}), the following result holds.\\
\begin{proposition}
For the BCD-GEM algorithm, $-n^{-1} \ell^{(0)}_{\mathrm{pen},\lambda}(\theta^{(m)})$ decreases
monotonically to some value $\bar{\ell}>-\infty$.
\end{proposition}
 In Remark~\ref{remark:remark}, we noted that, for the convex penalty function with $\gamma=0$, the M-Step of the
algorithm corresponds to exact coordinate-wise minimization of
$\mathop{Q_{\mathrm{pen}}}(\theta|\theta^{(m)})$. In this case, convergence
to a stationary point can be shown. \\
 \begin{theorem}\label{th.BCD-GEM}
Consider the BCD-GEM algorithm for the criterion function in
(\ref{eq:plik1}) with $\gamma = 0$. 
 Then, every cluster point $\bar{\theta}\in\Theta$ of the sequence $\{\theta^{(m)};m =
 0,1,2,\ldots\}$, generated by the BCD-GEM algorithm, is a stationary point
 of the criterion function in (\ref{eq:plik1}).
 \end{theorem}
A precise definition of a stationary point in a non-differentiable setup
and a proof of the Theorem are given in Appendix \ref{app2}. The proof uses the crucial facts that
$\mathop{Q_{\mathrm{pen}}}(\theta|\theta')$ is a convex function in $\theta$ and
that it is strictly convex in each coordinate of $\theta$.  


\section[Simulation and real data example]{Simulations, real data example and computational
  timings}\label{sec.numeric}
\subsection{Simulations}\label{subsec.simul}
We consider four different simulation setups. Simulation scenario 1 compares
the performance of the unpenalized MLE with our estimators from Section
\ref{subsec.MLEmixpen} (FMRLasso) and Section \ref{subsec.adaptive}
(FMRAdapt) in a  
situation where the total number of noise covariates grows
successively. For computing the unpenalized MLE, we used the
\texttt{R}-package FlexMix \citep{flexmix1,flexmix2,flexmix3}; 
Simulation 2 explores sparsity; Simulation 3 compares
cross-validation and BIC; and Simulation 4 compares the different penalty
functions with the parameters $\gamma=0, 1/2, 1$. For every setting, the
results are based on 100 independent simulation runs. 

All simulations are based on Gaussian FMR models as in (\ref{mod.mix});
the coefficients  
$\pi_{r}, \beta_{r}, \sigma_{r}$ and the sample size $n$ are specified below.
The covariate $X$ is generated from a multivariate
normal distribution with mean 0 and covariance structure as specified below.

Unless otherwise specified, the penalty with $\gamma=1$ is used
in all 
simulations. As explored empirically in
Simulation 4, in case of balanced problems
(approximately equal $\pi_r$), the FMRLasso performs similarly for all
three penalties.  In unbalanced
situations, the best results are typically achieved with $\gamma=1$.
In addition, unless otherwise specified, the true number of components $k$ is
assumed to 
be known. 

For all models, training-, validation- and test data are
generated of equal size $n$. The estimators are computed on the training
data, with the tuning parameter (e.g., $\lambda$) selected by minimizing twice the negative log-likelihood (log-likelihood loss) on the validation data. 
As performance measure, the predictive log-likelihood loss (twice the negative
log-likelihood) of the selected model is computed on the test data. 

Regarding variable selection, we count a covariable $X^{(j)}$ as selected
if $\hat{\beta}_{r,j}\neq0$ for at least one $r\in\{1,\ldots,k\}$. To
assess the performance of FMRLasso on 
recovering the sparsity structure, we report the number of truly  selected
covariates (True Positives) and falsely selected covariates (False
Positives). 

Obviously, the performances depend on the signal-to-noise ratio (SNR) which
we define for an FMR model as
\begin{eqnarray*}
\mathrm{SNR}=\frac{\Var(Y)}{\Var(Y|\beta_r=0; r=1,\ldots,k)}
=
\frac{\sum_{r=1}^{k}\pi_{r}(\beta_{r}^{T}\Cov(X)\beta_{r}+\sigma_{r}^{2})}{\sum_{r=1}^{k}\pi_{r}\sigma_{r}^{2}},
\end{eqnarray*}
where the last equality follows since $\EE[X] = 0$. 

\subsubsection{Simulation 1}\label{sec:sim1}
We consider five different FMR models: M1, M2, M3, M4 and M5. The parameters
$(\pi_{r},\beta_{r},\sigma_{r})$, the sample size $n$ of the training-,
validation- and test-data, the correlation structure of covariates
$\mathrm{corr}(X^{(l)},X^{(m)})$ and the signal-to-noise ratio (SNR) are
specified in Table~\ref{tab:sim1}. Models M1, M2, M3 and M5 have two components and five active covariates, whereas model M4 has three
components and six active covariates. M1, M2 and M3 differ only in
their variances $\sigma_{1}^2$, $\sigma_{2}^2$ and hence have different
signal-to-noise ratios. Model M5 has a non-diagonal
covariance structure. Furthermore, in model M5, the variances $\sigma_{1}^2$,
$\sigma_{2}^2$ are tuned to achieve the same signal-to-noise ratio as in model M1.

We compare the performances of the maximum likelihood estimator (MLE), the
FMRLasso and the FMRAdapt in a situation where the number of noise
covariates grows successively. For the models M1, M2, M3, M5 with two
components, we start with $p_{\mathrm{tot}}=5$ (no noise covariates) and go up to
$p_{\mathrm{tot}}=125$ (120 noise covariates). For the three component model M4, we
start with $p_{\mathrm{tot}}=6$ (no noise covariates) and go up to $p_{\mathrm{tot}}=155$
(149 noise covariates). 

The boxplots in Figures~\ref{fig:m1} - \ref{fig:m5} of the predictive
log-likelihood loss, denoted by \emph{Error}, the True Positives
(\emph{TP}) and the False Positives (\emph{FP}) over 
100 simulation runs summarize the results for the different models. We read
off from these boxplots that the MLE performs very badly when we add noise
covariates. On the other hand, our penalized estimators remain stable. For
example, for M1 the MLE with $p_{\mathrm{tot}}=20$ performs worse than the FMRLasso
with $p_{\mathrm{tot}}=125$, or for M4 the MLE with $p_{\mathrm{tot}}=10$ performs worse than
the FMRLasso with $p_{\mathrm{tot}}=75$. Impressive is also the huge gain of the
FMRAdapt method over FMRLasso in terms of log-likelihood loss and false
positives. 

\begin{table}[!h]
  \centering
\scriptsize
  \begin{tabular}{c|ccccc}
    &M1&M2&M3&M4&M5\\\hline
    n&100&100&100&150&100\\
    $\beta_1$&(3,3,3,3,3)&(3,3,3,3,3)&(3,3,3,3,3)&(3,3,0,0,0,0)&(3,3,3,3,3)\\
    $\beta_2$&(-1,-1,-1,-1,-1)&(-1,-1,-1,-1,-1)&(-1,-1,-1,-1,-1)&(0,0,-2,-2,0,0)&(-1,-1,-1,-1,-1)\\
    $\beta_3$&-&-&-&(0,0,0,0,-3,2)&-\\
    $\sigma$&0.5, 0.5&1, 1&1.5, 1.5&0.5, 0.5, 0.5&0.95, 0.95\\
    $\pi$&0.5, 0.5&0.5, 0.5&0.5, 0.5&1/3, 1/3, 1/3&0.5, 0.5\\
    $\mathrm{corr}(X^{(l)},X^{(m)})$&$\delta_{l,m}$&$\delta_{l,m}$&$\delta_{l,m}$&$\delta_{l,m}$&$0.8^{|l-m|}$\\
    SNR&101&26&12.1&53&101
  \end{tabular}
  \caption{Models for simulation 1. $\delta_{l,m}$ denotes Kronecker's delta.}
  \label{tab:sim1}
\end{table}

\begin{figure}[!h]
\begin{centering}
\includegraphics[angle=0,scale=0.41]{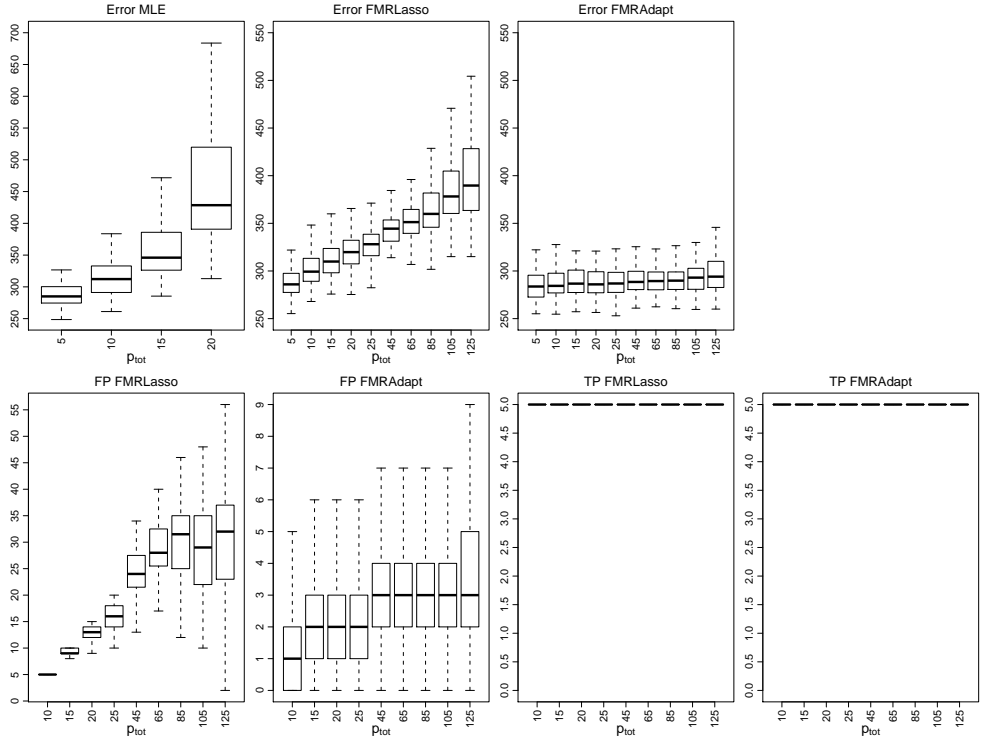}
\caption[Simulation 1, Model 1]{Simulation 1, Model M1. Top:
  Predictive log-likelihood loss (\emph{Error}) for MLE, FMRLasso,
  FMRAdapt. Bottom: False Positives (\emph{FP}) and True Positives
  (\emph{TP}) for FMRLasso and FMRAdapt.}  
\label{fig:m1}
\end{centering}
\end{figure}

\begin{figure}[!p]
\begin{centering}
\includegraphics[angle=0,scale=0.41]{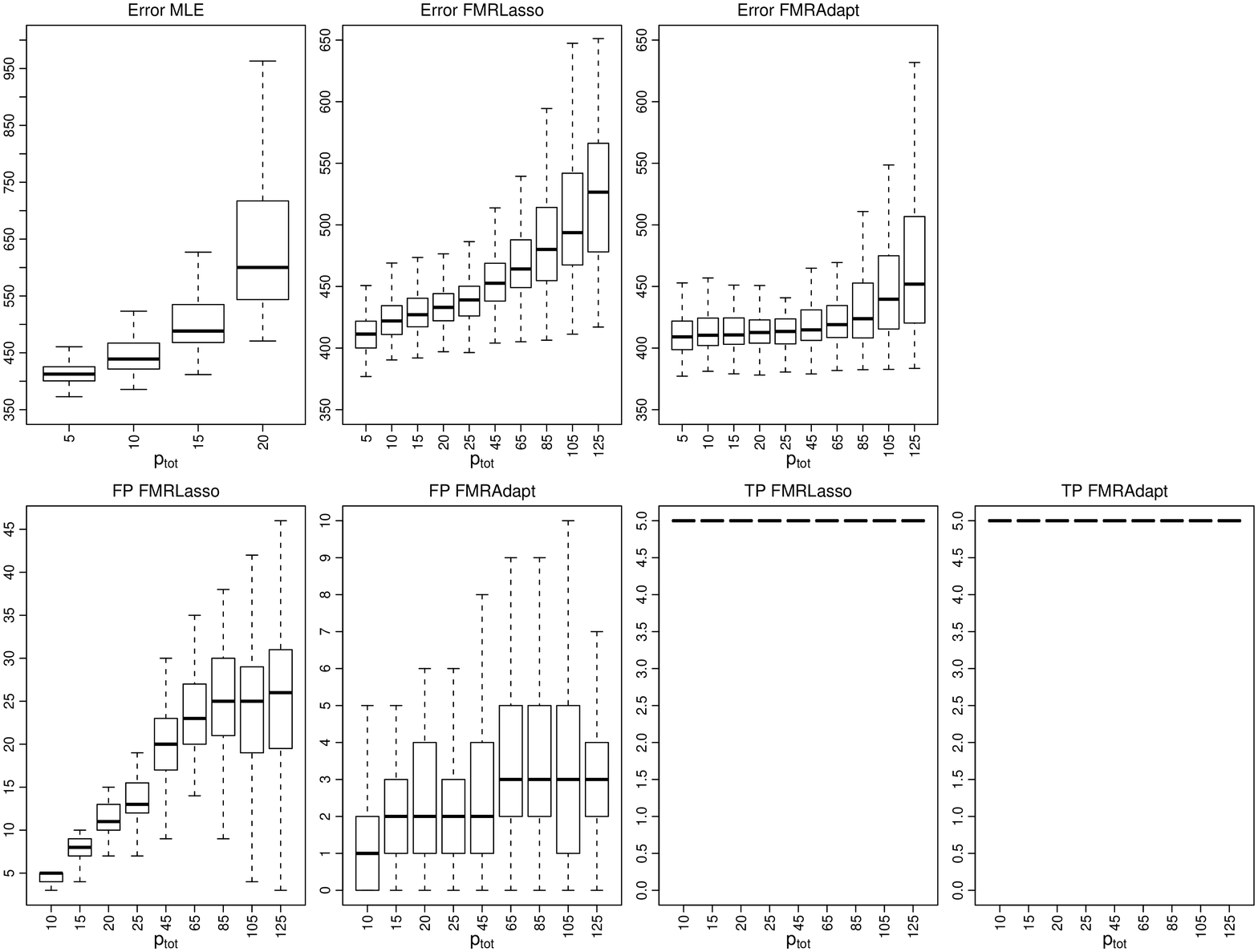}
  \caption[Simulation 1, Model 2]
  {Simulation 1, Model M2. Same notation as in Figure~\ref{fig:m1}.}
\label{fig:m2}
\end{centering}
\end{figure}

\begin{figure}[!p]
\begin{centering}
\includegraphics[angle=0,scale=0.41]{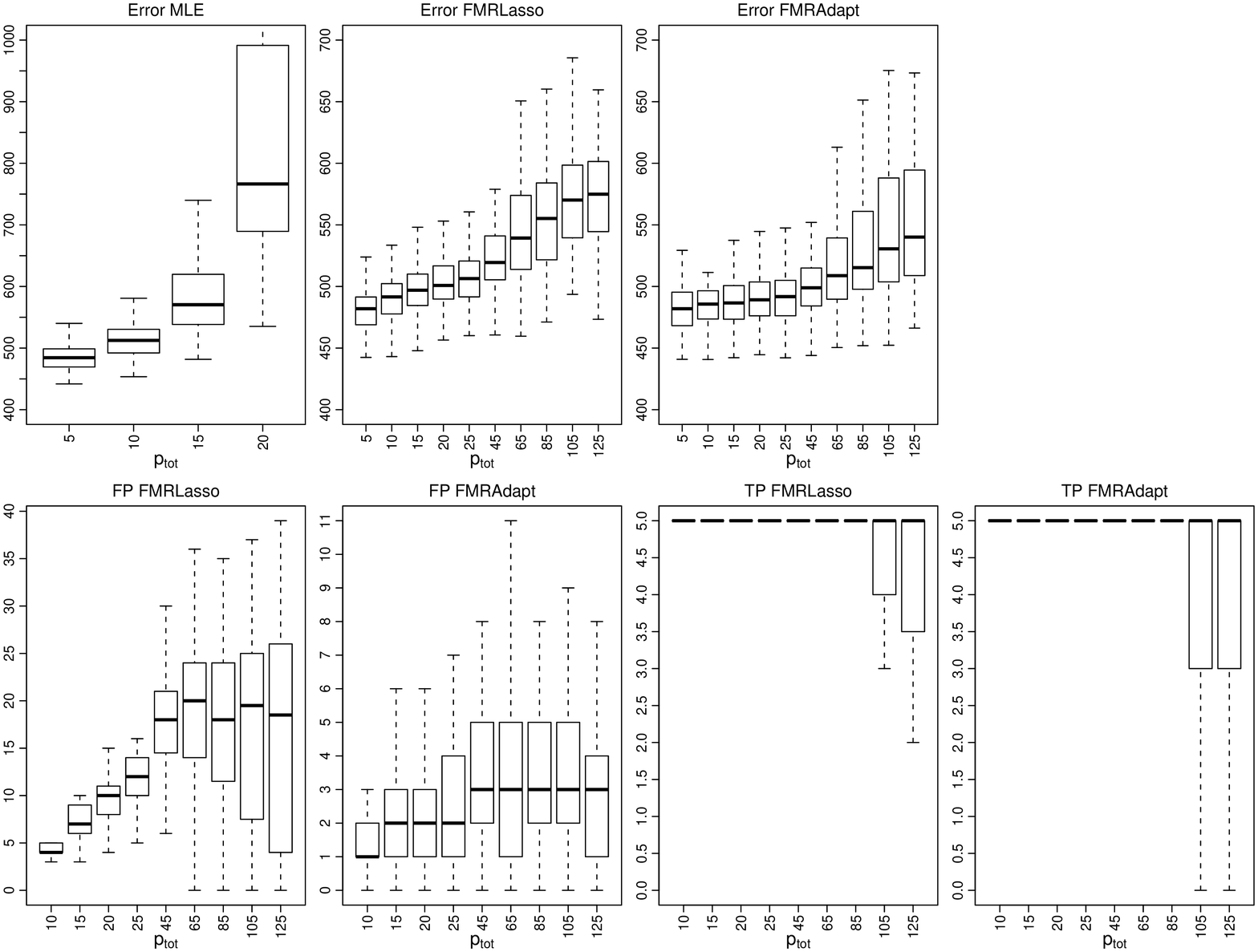}
  \caption[Simulation 1, Model 3]{Simulation 1, Model M3. Same notation as in Figure~\ref{fig:m1}.}
\label{fig:m3}
\end{centering}
\end{figure}

\begin{figure}[!p]
\begin{centering}
\includegraphics[angle=0,scale=0.41]{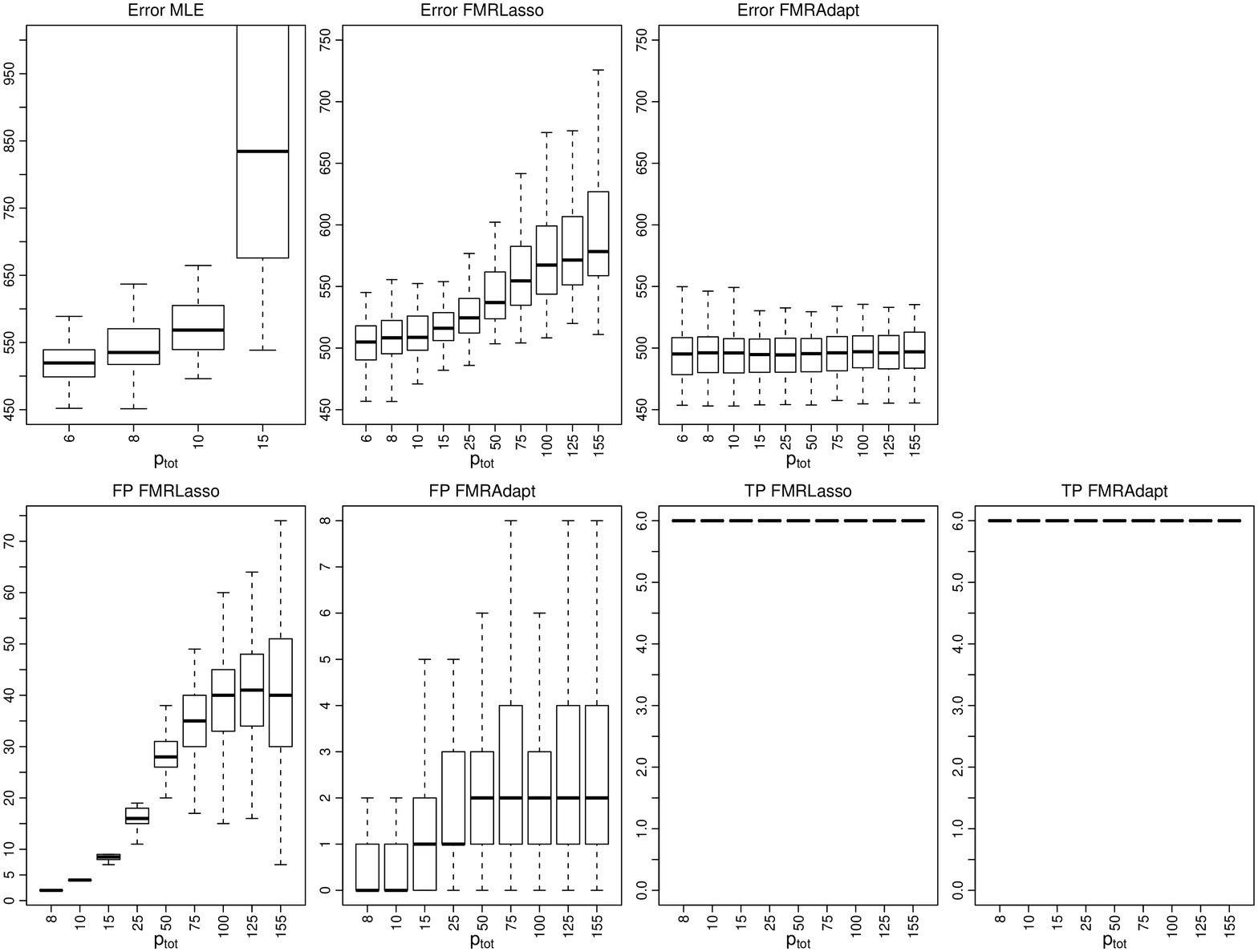}
  \caption[Simulation 1, Model 4]{Simulation 1, Model M4. Same notation as in Figure~\ref{fig:m1}.}
\label{fig:m4}
\end{centering}
\end{figure}

\begin{figure}[!p]\label{fig:m5}
\begin{centering}
\includegraphics[angle=0,scale=0.41]{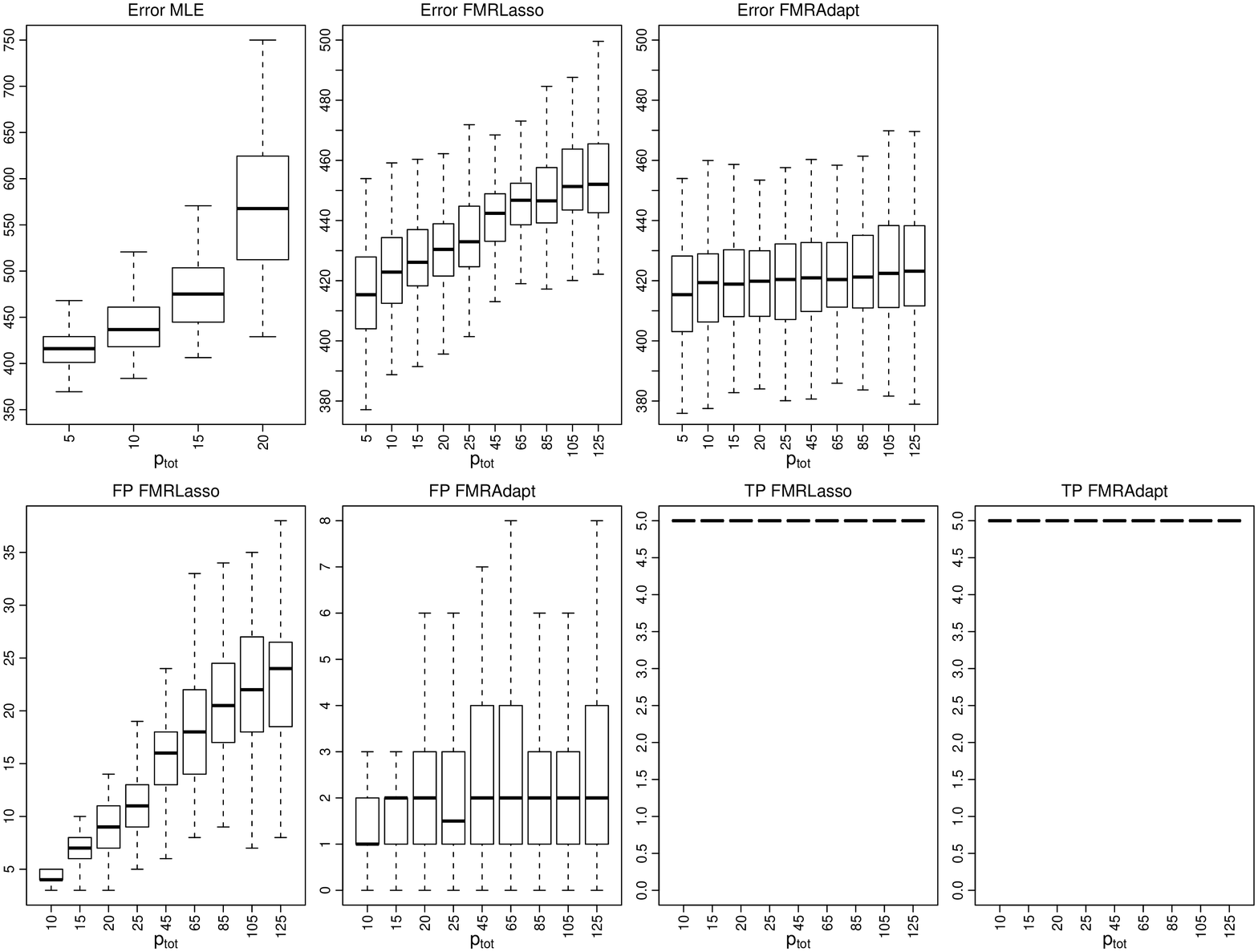}
  \caption[Simulation 1, Model 5]{Simulation 1, Model M5. Same notation as in Figure~\ref{fig:m1}.}
\label{fig:m5}
\end{centering}
\end{figure}

\subsubsection{Simulation 2}
In this section, we explore the sparsity properties of the FMRLasso. The model specifications are given in
Table~\ref{tab:sim2}. Consider the ratios $p_{\mathrm{act}}:n:p_{\mathrm{tot}}$. The total
number of covariates $p_{\mathrm{tot}}$ grows faster than the number of observations
$n$ and the number of active covariates $p_{\mathrm{act}}$: when $p_{\mathrm{tot}}$ is
doubled, $p_{\mathrm{act}}$ is 
raised by one and $n$ is raised by $50$ from model to model.
In particular, we obtain a series of models which gets ``sparser'' as $n$
grows (larger ratio $n/p_{\mathrm{act}}$). In order to compare the performance of
the FMRLasso, we report the True 
Positive Rate (\emph{TPR}) and the False Positive Rate (\emph{FPR}) defined as:
\begin{eqnarray*}
\mathrm{TPR}&=&\frac{\# \textrm{truly selected covariates}}{\#\textrm{active covariates}},\\
\mathrm{FPR}&=&\frac{\# \textrm{falsely selected covariates}}{\#\textrm{inactive covariates}}.
\end{eqnarray*}

These numbers are reported in Figure~\ref{fig:sim2}. We see that the False
Positive Rate approaches zero for sparser models, indicating that the
FMRLasso recovers the true model better in sparser settings regardless of
the large number of noise covariates.

\begin{table}[!h]
\small
  \centering
  \begin{tabular}{|c|ccccccc|}\hline
    $p_{act}$&3&4&5&6&7&8&9\\
    n&50&100&150&200&250&300&350\\
    $p_{tot}$&10&20&40&80&160&320&640\\
    $\beta_1$&&&&(3, 3, 3, 0, 0, \ldots)&&&\\
    $\beta_2$&&&&(-1, -1, -1, 0, 0, \ldots)&&&\\
    $\sigma$&&&&0.5, 0.5&&&\\
    $\pi$&&&&0.5, 0.5&&&\\\hline
  \end{tabular}
  \caption{Series of models for simulation 2 which gets ``sparser'' as $n$
    grows: when $p_{\mathrm{tot}}$ is doubled, $p_{\mathrm{act}}$ is
raised by one and $n$ is raised by $50$ from model to model.}
  \label{tab:sim2}
\end{table}

\begin{figure}[!h]
\begin{centering}
\includegraphics[scale=0.42]{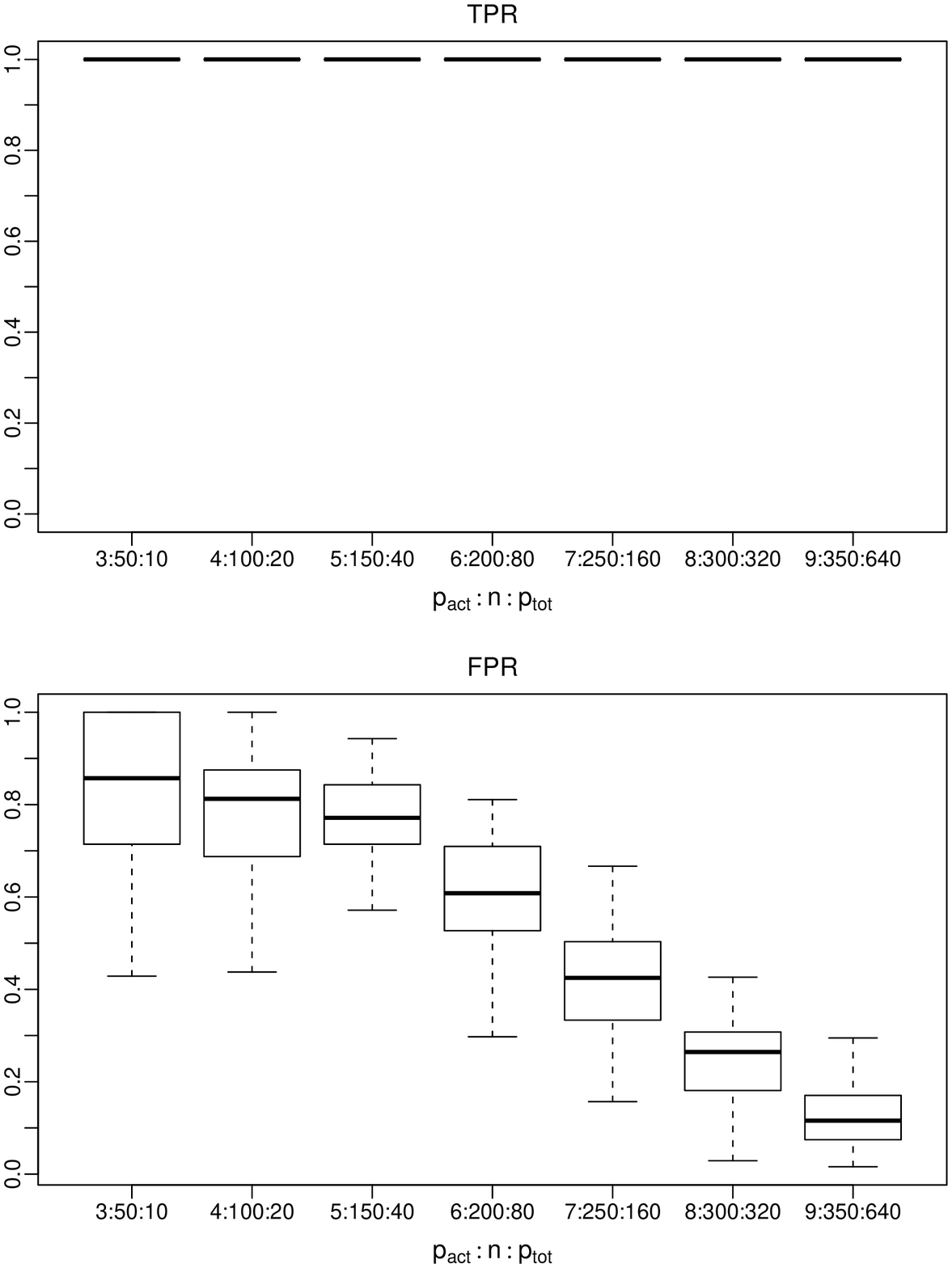} 
  \caption[Simulation 2]
  {Simulation 2 compares the performance of the FMRLasso for a series of models
    which gets ``sparser'' as the sample size grows. Top: True Positive Rate (\emph{TPR}). Bottom: False Positive Rate (\emph{FPR}) over 100 simulation runs.}
\label{fig:sim2}
\end{centering}
\end{figure}

\subsubsection{Simulation 3}
So far, we regarded the number $k$ of components as given, while we have
chosen an optimal $\lambda_{\mathrm{opt}}$ by minimizing the negative log-likelihood
loss on 
validation data. In this section, we compare the performance of
10-fold cross-validation and the BIC criterion presented in
Section~\ref{sec:select.tuning} for selecting the tuning
parameters $k$ and $\lambda$. We use model M1 of Section~\ref{sec:sim1} with $p_{\mathrm{tot}}=25, 50,
75$. For each of these models, we tune the FMRLasso estimator according to
the following strategies:
\begin{itemize}
\item[(1)] Assume the number of components is given ($k=2$). Choose the
  optimal tuning parameter $\lambda_{\mathrm{opt}}$ using 10-fold
  cross-validation. 
\item[(2)] Assume the number of components is given ($k=2$). Choose $\lambda_{\mathrm{opt}}$ by minimizing the BIC criterion (\ref{eq:bic}).
\item[(3)] Choose the number of components $k\in\{1,2,3\}$ and
  $\lambda_{\mathrm{opt}}$ by minimizing the BIC criterion (\ref{eq:bic}).
\end{itemize}
The results of this simulation are presented in Figure~\ref{fig:sim3}, where 
boxplots of the log-likelihood loss (\emph{Error}) are
shown. All three strategies perform equally well. With $p_{\mathrm{tot}}=25$ the BIC
criterion in strategy (3) always chooses $k=2$. For the model with $p_{\mathrm{tot}}=50$, strategy
(3) chooses $k=2$ in $98$ simulation runs and $k=3$ in two runs. Finally,
with $p_{\mathrm{tot}}=75$, the third strategy chooses $k=2$ in $92$ runs and $k=3$ eight times. 

\subsubsection{Simulation 4}
In the preceding simulations, we always used the value $\gamma=1$ in the
penalty term of the FMRLasso estimator (\ref{eq:plik1}). In this section, we
compare the FMRLasso for different values $\gamma=0, 1/2, 1$. First, we
compute the FMRLasso for $\gamma=0, 1/2, 1$ on model
M1 of Section~\ref{sec:sim1} with $p_{\mathrm{tot}}=50$. Then we do the same calculations for an ``unbalanced''
version of this model with $\pi_{1}=0.3$ and $\pi_{2}=0.7$.

In Figure~\ref{fig:sim4}, the boxplots of the log-likelihood loss (\emph{Error}), the
False Positives (\emph{FP}) and the True Positives (\emph{TP}) over 100 simulation runs
are shown. We see that the FMRLasso performs similarly for $\gamma=0, 1/2,
1$. Nevertheless, the value $\gamma=1$ is slightly preferable in the
``unbalanced'' setup.

\begin{figure}[!h]
\begin{centering}
\includegraphics[scale=0.5]{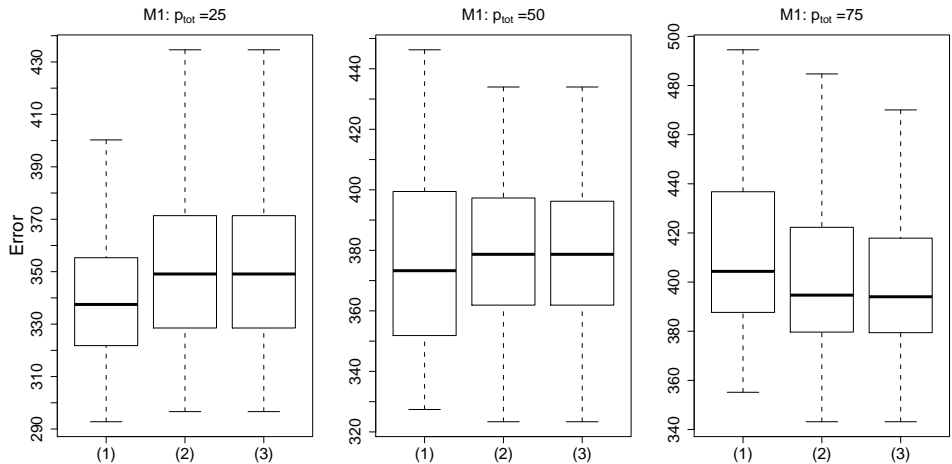} 
  \caption[Simulation 3]
  {Simulation 3 compares different strategies for choosing the tuning
    parameters $k$ and $\lambda$. The boxplots show the predictive
    log-likelihood loss (\emph{Error}) of the FMRLasso, tuned by strategies
    (1), (2) and (3),
    for model M1 with $p_{\mathrm{tot}}=25, 50, 75$.}
\label{fig:sim3}
\end{centering}
\end{figure}

\begin{figure}[htbp]
\begin{centering}
\includegraphics[scale=0.5]{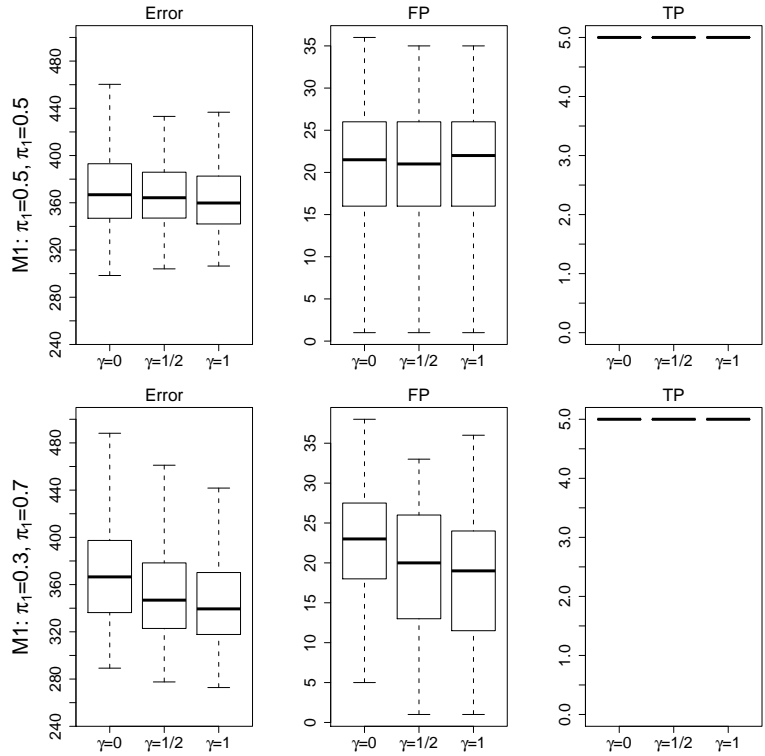} 
  \caption[Simulation 4]
  {Simulation 4 compares the FMRLasso for different values $\gamma=0,
    1/2, 1$. The upper row of panels shows the boxplots of the log-likelihood loss (\emph{Error}),
    the False Positives (\emph{FP}) and the True Positives (\emph{TP}) for model M1 with
    $p_{\mathrm{tot}}=50$ and $\pi_{1}=\pi_{2}=0.5$. The lower row of panels shows
    the same boxplots for an ``unbalanced'' version of model M1 with
    $\pi_{1}=0.3$ and $\pi_{2}=0.7$.}
\label{fig:sim4}
\end{centering}
\end{figure}
\clearpage
\subsection{Real data example}\label{subsec.riboflavin}
We now apply the FMRLasso to a dataset of riboflavin (vitamin $B_{2}$)
production by \emph{Bacillus Subtilis}. The real-valued response variable
is the logarithm of the riboflavin production rate. The data has been
kindly provided by DSM (Switzerland). There are $p=4088$
covariates (genes) measuring the logarithm of the expression level of 4088
genes and measurements of $n=146$ genetically engineered mutants of
\emph{Bacillus Subtilis}. The population seems to be rather heterogeneous as
there are different
strains of \emph{Bacillus Subtilis} which are cultured
under different fermentation conditions. We do not know the different
homogeneity subgroups. For this reason, a FMR model with more
than one component might be more appropriate than a simple linear
regression model.

We compute the FMRLasso estimator for $k=1,\ldots,5$ components. To keep the
computational effort reasonable, we use only the 100
covariates (genes) exhibiting the highest empirical variances. We choose the
optimal tuning 
parameter $\lambda_{\mathrm{opt}}$ by 10-fold cross-validation (using the
log-likelihood loss). As a result, we get five different estimators which we compare
according to their cross-validated log-likelihood loss (\emph{CV Error}). These
numbers are plotted in Figure~\ref{fig:realdata1}. The estimator with three
components performs clearly best, resulting in a 17\% improvement in
prediction over a
(non-mixture) linear model, and it selects 51 genes. In Figure~\ref{fig:realdata2}, the coefficients
of the $20$ most important genes, ordered according to
$\sum_{r=1}^{3}|\hat{\beta}_{r,j}|$, are shown. From the important
variables, only gene 83 shows the opposite sign of the estimated regression coefficients among the
three different mixture components. However, it happens that some
covariates (genes) exhibit a strong effect in one or two mixture components
but none in the remaining other components. 
Finally, for comparison, the one-component (non-mixture) model selects 26
genes, of which 24 are also selected in the
three-component model. 
 
\begin{figure}[htbp]
\begin{centering}
\includegraphics[scale=0.43]{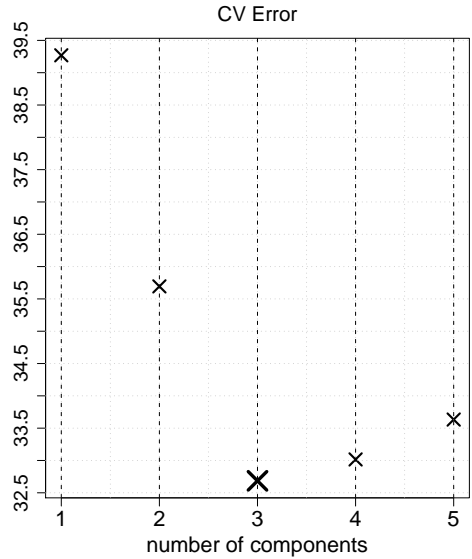}
  \caption[Real data example]
  {Riboflavin production data. Cross-validated negative log-likelihood loss
    (\emph{CV Error}) for the 
    FMRLasso estimator when varying  over different numbers of
    components.}
\label{fig:realdata1}
\end{centering}
\end{figure}

\begin{figure}[htbp]
\begin{centering}
\includegraphics[scale=0.45]{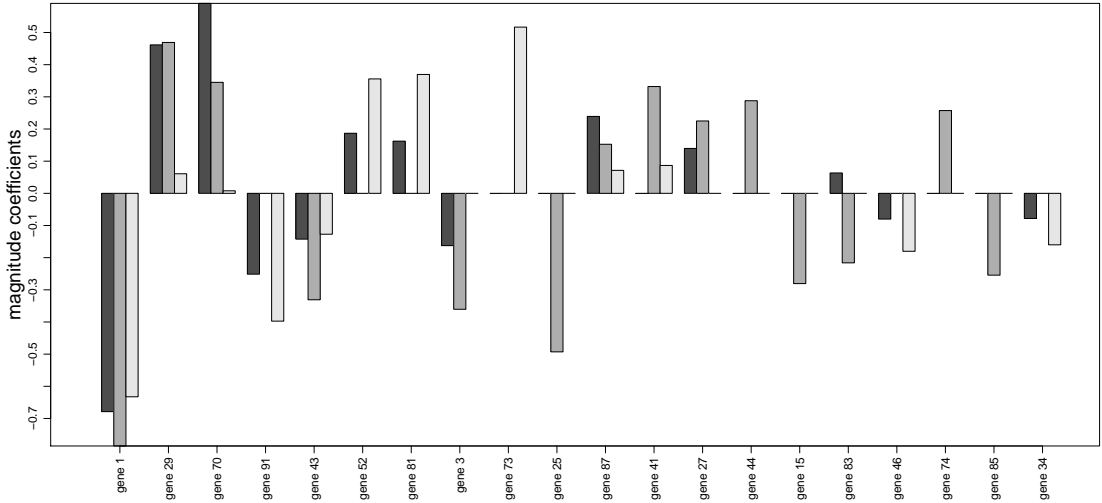}
  \caption[Real data example]
  {Riboflavin production data. Coefficients of the $20$ most important
    genes, ordered according to 
    $\sum_{r=1}^{3}|\hat{\beta}_{r,j}|$, for the prediction optimal model with
    three components.}
\label{fig:realdata2}
\end{centering}
\end{figure}

\subsection{Computational timings}\label{subsec.comptiming}
In this section, we report on the run times of the BCD-GEM algorithm on two
high-dimensional examples. In particular, we focus on the substantial gain of
speed achieved by using the active set version of the algorithm described
in Section~\ref{subsec.activeset}. All computations were carried out with
the statistical computing language and environment \texttt{R}. Timings
depend on the stopping criterion used in 
the algorithm. We stop the algorithm if the relative function
improvement and the relative change of the parameter vector are small
enough, i.e., 
\begin{align*}
\frac{|\ell_{\mathrm{pen},\lambda}^{(\gamma)}(\theta^{(m+1)})-\ell_{\mathrm{pen},\lambda}^{(\gamma)}(\theta^{(m)})|}{1+|\ell_{\mathrm{pen},\lambda}^{(\gamma)}(\theta^{(m+1)})|} 
\leq \tau,
\quad&\max_{j}\left\{\frac{|\theta_j^{(m+1)}-\theta_j^{(m)}|}{1+|\theta_j^{(m+1)}|}\right\}\leq\sqrt{\tau}, 
\end{align*}
with $\tau=10^{-6}.$

We consider a high-dimensional version
of the two component model M1 from Section~\ref{sec:sim1} with $n=200$,
$p_{\mathrm{tot}}=1000$ and the riboflavin dataset from
Section~\ref{subsec.riboflavin} with three components, $n=146$ and
$p_{\mathrm{tot}}=100$. We use the BCD-GEM algorithm with and without
active set strategy to fit the FMRLasso on a small grid of eight values for
$\lambda$. The corresponding BIC, CPU times (in
seconds) and number of EM-iterations are reported in Tables~\ref{tab:cpu.model1} and
\ref{tab:cpu.realdata}. The values for the BCD-GEM without active set
strategy are written in brackets. For model M1 and an appropriate $\lambda$
with minimal BIC score, the 
active set algorithm converges in 5.96 seconds whereas the standard
BCD-GEM needs 53.15 seconds. There is also a considerable gain of speed for
the real data: 
0.89 seconds versus 3.57 seconds for $\lambda$ with optimal BIC. Note that
in Table \ref{tab:cpu.model1}, the BIC scores sometimes differ
substantially for inappropriate values of $\lambda$. For such
regularization parameters, the solutions are unstable and different
local optima are attained depending on the algorithm used. However, if the
regularization parameter is in a reasonable range with low BIC score, the
results stabilize. 
\begin{table}[!hp]
{\centering
  \tabcolsep 1.8pt
    \scriptsize
\begin{tabular}{|@{}c|cccccccc@{}|}\hline
    $\lambda$&10.0&15.6&21.1&26.7&32.2&37.8&43.3&48.9\\\hline
    BIC&2033 (2022)&1606 (1748)&951 (959)&941 (940)&989 (983)&1236 (1073)&1214 (1216)&
    1206 (1203)\\
    CPU [s]&26.78 (269.91)&17.05 (165.78)&8.63 (82.78)&5.96
    (53.15)&5.08 (44.23)&4.23 (37.27)&3.35 (18.99)&3.30 (15.62)\\
    $\#$ EM-iter.&277.0 (341.5)&196.0 (205.0)&96.0 (100.5)&63.5 (64.5)&56.0 (53.5)&41.5
    (46.0)&31.5 (23.0)&25.0 (19.0)\\\hline
  \end{tabular}
\caption{Model M1 with $n=200$ and $p_{\mathrm{tot}}=1000$. Median over 10
    simulation runs of BIC, CPU times and number of EM-iterations for the BCD-GEM with and without active set
    strategy (the latter in brackets).}\label{tab:cpu.model1}
}
\vspace{2.0cm}
{\centering
  \tabcolsep 4.6 pt
  \scriptsize
  \begin{tabular}{|@{}c|cccccccc@{}|}\hline
    $\lambda$&3.0&13.8&24.6&35.4&46.2&57.0&67.8&78.6\\\hline
    BIC&560 (628)&536 (530)&516 (522)& 532 (525)& 541 (540)&561 (580)&592 (591)&
    611 (613)\\
    CPU [s]&22.40 (29.98)&1.35 (3.28)&0.89 (3.57)&0.86
    (3.34)&0.78 (3.87)&0.69 (2.42)&0.37 (2.56)&0.85 (4.05)\\
    $\#$ EM-iter.&3389 (2078)&345 (239)&287 (266)&298 (247)&296 (290)&248
    (184)&129 (192)&313 (302)\\\hline
  \end{tabular}
\caption{Riboflavin data with $k=3$, $n=146$ and $p_{\mathrm{tot}}=100$. BIC, CPU
    times and number of EM-iterations for the BCD-GEM with and without active set
    strategy (the latter in brackets).}\label{tab:cpu.realdata}
}
\end{table}

\section{Discussion}
We have presented an $\ell_1$-penalized estimator for a finite mixture of
high-dimensional Gaussian regressions where the number of covariates may
greatly exceed sample size. Such a model and the corresponding Lasso-type
estimator are useful to blindly account for often encountered inhomogeneity
of high-dimensional data. On a high-dimensional real data example, we
demonstrate a 17\% gain in prediction accuracy over a (non-mixture) linear
model.  

The computation and mathematical analysis in such a high-dim\-en\-sion\-al
mixture model is challenging due to the non-convex behavior of the negative
log-likelihood. Moreover, with high-dimensional estimation defined via
optimization of a non-convex objective function, there is a major gap
between the actual computation and the procedure analyzed in theory. We do
not provide an answer to this issue in this thesis. Regarding the
computation in FMR models, a simple 
reparameterization is 
very beneficial and the $\ell_1$-penalty term makes 
the optimization problem numerically much better behaved. We develop an
efficient generalized EM algorithm and we prove its numerical convergence to a
stationary point. Regarding the statistical properties, besides standard
low-dimensional asymptotics, we present a non-asymptotic oracle inequality
for the 
Lasso-type estimator in a high-dimensional setting with general, non-convex
but smooth loss 
functions. The mathematical arguments are different than what is typically
used for convex losses.

\begin{appendices} 
\section{Proofs for Section \ref{sec.asympt}}\label{app0} 
\subsection{Proof of Theorem \ref{theorem:rootn}}
We assume the regularity assumptions (A)-(C) of \citet{fanli}. The theorem follows from Theorem 1 of \cite{fanli}.
\hfill $\sqcup \mkern -12mu \sqcap$

\subsection{Proof of Theorem \ref{theorem:oracle}}
In order to keep the notation simple, we give the proof for a two class
mixture with $k=2$. All arguments in the proof can also be used for a general
mixture with more than two components. Remember that
$-n^{-1}\ell_{\mathrm{adapt}}(\theta)$ is given by
\begin{eqnarray*}
-n^{-1}\ell_{\mathrm{adapt}}(\theta)&=&-n^{-1}\ell(\theta)+\lambda\Big(\pi_{1}^{\gamma}\sum_{j=1}^{p}w_{1,j}|\phi_{1,j}|+(1-\pi_{1})^{\gamma}\sum_{j=1}^{p}w_{2,j}|\phi_{2,j}|\Big),
\end{eqnarray*}
where
$\ell(\theta)$ is the log-likelihood function. The weights
$w_{r,j}$ are given by $w_{r,j}=\frac{1}{|\phi_{r,j}^{\mathrm{ini}}|}$, $r=1,2$, and
$j=1,\ldots,p$.\\ \\
Assertion 1.\\
Let $\hat{\theta}$ be a root-$n$ consistent local minimizer of
$-n^{-1}\ell_{\mathrm{adapt}}(\theta)$ (construction as in \cite{fanli}). 

For all $(r,j) \in S$, we easily see from consistency of $\hat{\theta}$ that $\Prob[(r,j) \in \hat{S}]\rightarrow 1.$ It then remains to
show that for all $(r,j)\in S^{c}$, $\Prob[(r,j)\in\hat{S}^{c}]\rightarrow1.$
Assume the contrary, i.e., w.l.o.g there is an $s\in\{1,\ldots,p\}$ with $\phi_{1,s}=0$ such that $\hat{\phi}_{1,s}\neq0$ with non-vanishing probability.

By Taylor's theorem, applied to the function $n^{-1}\frac{\partial\ell(\theta)}{\partial
  \phi_{1,s}}$, there exists a (random) vector $\tilde{\theta}$
on the line segment between $\theta_{0}$ and $\hat{\theta}$ such that
\begin{align*}
\frac{1}{n}\frac{\partial\ell_{\mathrm{adapt}}}{\partial
  \phi_{1,s}}\Big|_{\hat{\theta}}=&\underbrace{\frac{1}{n}\frac{\partial\ell}{\partial
  \phi_{1,s}}\Big|_{\theta_{0}}}_{(1)}+\underbrace{\frac{1}{n}\frac{\partial\ell'}{\partial
  \phi_{1,s}}\Big|_{\theta_{0}}}_{(2)}\left(\hat{\theta}-\theta_{0}\right)+\\
&\frac{1}{2}\left(\hat{\theta}-\theta_{0}\right)^{T}\underbrace{\frac{1}{n}\frac{\partial\ell''}{\partial
  \phi_{1,s}}\Big|_{\tilde{\theta}}}_{(3)}\left(\hat{\theta}-\theta_{0}\right)-
\lambda \hat{\pi}_1^{\gamma}
w_{1,s}\mathrm{sgn}(\hat{\phi}_{1,s}).
\end{align*}
Now, using the regularity assumptions and the central limit theorem, term (1) is of order $O_{P}(\frac{1}{\sqrt{n}})$. Similarly, term (2) is of
order $O_{P}(1)$ by the law of large numbers. Term (3) is of order
$O_{P}(1)$ by the law of large numbers and the regularity
condition on the third derivatives (condition (C) of \citet{fanli}). 
Therefore, we have
\begin{eqnarray*}
&&\frac{1}{n}\frac{\partial\ell_{\mathrm{adapt}}}{\partial\phi_{1,s}}\Big|_{\hat{\theta}}=O_{P}(\frac{1}{\sqrt{n}})+(O_{P}(1)+
  (\hat{\theta}-\theta_{0})^{T}O_{P}(1))(\hat{\theta}-\theta_{0}) 
-\lambda \hat{\pi}_1^{\gamma}
w_{1,s}\mathrm{sgn}(\hat{\phi}_{1,s}).
\end{eqnarray*}
As $\hat{\theta}$ is root-$n$ consistent we get 
\begin{eqnarray}\label{eq:proofthm2}
\frac{1}{n}\frac{\partial\ell_{\mathrm{adapt}}}{\partial\phi_{1,s}}\Big|_{\hat{\theta}}&=&O_{P}(\frac{1}{\sqrt{n}})+\left(O_{P}(1)+o_{P}(1)O_{P}(1)\right)O_{P}(\frac{1}{\sqrt{n}})- \lambda \hat{\pi}_1^{\gamma}
w_{1,s}\mathrm{sgn}(\hat{\phi}_{1,s})\nonumber\\ 
&=&\frac{1}{\sqrt{n}} \left(O_{P}(1) - 
  \frac{n \lambda}{\sqrt{n}}\hat{\pi}_1^{\gamma} w_{1,s}\mathrm{sgn}(\hat{\phi}_{1,s})\right).
\end{eqnarray}
From the assumption on the initial estimator, we
have \[\frac{n \lambda}{\sqrt{n}}
w_{1,s}=\frac{n \lambda}{\sqrt{n}|\phi_{1,s}^{\mathrm{ini}}|}=\frac{n
  \lambda}{O_{P}(1)}\rightarrow 
\infty \qquad \textrm{as} \qquad n \lambda \rightarrow \infty.\] 
Therefore, the second term in the brackets of (\ref{eq:proofthm2}) dominates the first and the
probability of the event
\[\left\{\mathrm{sgn}\left(\frac{1}{n}\frac{\partial\ell_{\mathrm{adapt}}}{\partial\phi_{1,s}}\Big|_{\hat{\theta}}\right)=
  - \mathrm{sgn}(\hat{\phi}_{1,s})\neq0\right\}\]
tends to 1. But this contradicts the assumption that $\hat{\theta}$ is a
local minimizer
(i.e., $\frac{1}{n}\frac{\partial\ell_{\mathrm{adapt}}}{\partial\phi_{1,s}}\big|_{\hat{\theta}}=~0$).\\ \\
Assertion 2.\\
Write
$\theta=(\theta_{S},\theta_{S^{c}})$. From part 1), it
follows that with probability tending to one $\hat{\theta}_{S}$ is a root-$n$ local minimizer of $-n^{-1}\ell_{\mathrm{adapt}}\left(\theta_{S},0\right) $.
By using a Taylor expansion we find, 
\begin{align*}
0=&\frac{1}{n}\ell'_{\mathrm{adapt}}|_{\hat{\theta}_{S}}\\
=&\frac{1}{n}
\ell'|_{\theta_{0,S}}+\underbrace{\frac{1}{n}
  \ell''|_{\theta_{0,S}}}_{(1)}\big(\hat{\theta}_{S}-\theta_{0,S}\big)+\frac{1}{2}\underbrace{\big(\hat{\theta}_{S}-\theta_{0,S}\big)^{T}}_{(2)}\underbrace{\frac{1}{n}
  \ell'''|_{\tilde{\theta}_{S}}}_{(3)}\big(\hat{\theta}_{S}-\theta_{0,S}\big)\\
&- \lambda \left(\!\begin{array}{c}
\hat{\pi}_1^{\gamma}w_{1,S}\,\mathrm{sgn}(\hat{\phi}_{1,S})\\
(1-\hat{\pi}_1)^{\gamma}w_{2,S}\,\mathrm{sgn}(\hat{\phi}_{2,S})\\
0\\
0\\
\gamma\hat{\pi}_1^{\gamma-1}\!\sum\limits_{(1,j)\in S}\!
w_{1,j}|\hat{\phi}_{1,j}|-\gamma(1-\hat{\pi}_1)^{\gamma-1}\!\sum\limits_{(2,j)\in S}\!
w_{2,j}|\hat{\phi}_{2,j}|
\end{array}\!\right).
\end{align*}
Now term (1) is of order $-I_{S}(\theta_{0})+o_{P}(1)$ (law of large numbers); term~(2) is of
order $o_{P}(1)$ (consistency); and term (3), with some abuse of
notation an $(|S|+3)$-vector of $(|S|+3)\times(|S|+3)$ matrices, is of order $O_{P}(1)$ (law of
large numbers and regularity condition on the third derivatives). Therefore, we have
\[\sqrt{n}\frac{1}{n}
\ell'|_{\theta_{0,S}}+\left(-I_{S}(\theta_{0})+o_{P}(1)\right)\sqrt{n}(\hat{\theta}_{S}-\theta_{0,S})-\sqrt{n}\lambda
O_{P}(1)=0,
\]
or
\begin{eqnarray}\label{eq:oracle}
\left(-I_{S}(\theta_{0})+o_{P}(1)\right)\sqrt{n}(\hat{\theta}_{S}-\theta_{0,S})-
\sqrt{n} \lambda O_{P}(1)=-\frac{1}{\sqrt{n}}\ell'|_{\theta_{0,S}}.
\end{eqnarray}
Notice that $\frac{1}{\sqrt{n}}\ell'|_{\theta_{0,S}}\leadsto^{d}\calN(0,I_{S}(\theta_{0}))$
by the central limit theorem. Furthermore, $\sqrt{n} \lambda =o(1)$ as
 $\lambda=o(n^{-1/2})$. Therefore, 
$$\sqrt{n}(\hat{\theta}_{S}-\theta_{0,S})\leadsto^{d}
\calN(0,I_{S}(\theta_{0})^{-1})$$
 follows from Equation (\ref{eq:oracle}). \hfill $\sqcup \mkern -12mu \sqcap$

\section{Proofs for Section \ref{subsec.ashighdim}}\label{app1}
\subsection{Proof of Lemma \ref{marginlemma}}
Using a Taylor expansion,
$${\cal E} (\psi \vert \psi_0) =
(\psi- \psi_0)^T I( \psi_0)  (\psi - \psi_0) / 2 + r_{\psi} , $$
where
$$ |r_{\psi} | \le {\| \psi - \psi_0 \|_1^3 \over 6}
\int\sup_{\psi \in \Psi}  \max_{j_1 , j_2 , j_3 } 
\biggl | {\partial^3   l_{\psi} \over \partial \psi_{j_1} \partial
\psi_{j_2} \partial \psi_{j_3} }\biggr  | f_{\psi_0} d \mu $$
$$ \le { d^{3/2} C_3 \over 6} \| \psi - \psi_0 \|_2^3 . $$
Hence
$${\cal E} (\psi \vert \psi_0(x) )  \ge \| \psi -\psi_0 (x) \|_2^2
\Lambda_{\rm min}^2 / 2 - d^{3/2} C_3 \| \psi - \psi_0 (x) \|_2^3 /6 . $$
Now, apply the auxiliary lemma below, with $K_0^2 = d K^2 $, 
$\Lambda^2 = \Lambda_{\rm min}^2 / 2$, and \\
$C= d^{3/2} C_3 / 6 $.
\hfill $\sqcup \mkern -12mu \sqcap$

{\em Auxiliary Lemma.} {\it Let $h : [-K_0 , K_0 ] \rightarrow [0, \infty) $ have the following
properties:

(i) $\forall \ \eps >0 $ $\exists$ $\alpha_{\eps} > 0$ such that
$\inf_{\eps < |z| \le K_0 } h(z) \ge \alpha_{\eps} $,

(ii) $\exists$ $\Lambda >0$, $C>0$, such that $\forall \ |z| \le K_0 $,
$h(z) \ge
\Lambda^2 z^2 - C | z|^3 $.

Then $\forall \ |z| \le K_0$, 
$$h(z) \ge z^2 / C_0^2 ,$$
where 
$$C_0^2 = \max \biggl [ {1 \over \eps_0 } , { K_0^2 \over \alpha_{\eps_0}} \biggr ] , \ \eps_0 = {\Lambda^2 \over  2 C} .$$}

\emph{Proof (Auxiliary Lemma)}

If $\eps_0 > K_0$, we have 
$h(z) \ge \Lambda^2 z^2 / 2 $ for  all $| z | \le K_0 $.

If $\eps_0 \le K_0 $ and $|z| \le \eps_0$, we also have
$h(z) \ge( \Lambda^2 - \eps_0 C) z^2 \ge \Lambda^2 z^2 / 2 $.

If $\eps_0 \le K_0$ and $\eps_0 < |z| \le K_0$, we have
$h(z) \ge \alpha_{\eps_0} =K_0^2  \alpha_{\eps_0} / K_0^2 \ge
|z|^2 \alpha_{\eps_0} / K_0^2 $.\hfill $\sqcup \mkern -12mu \sqcap$

\subsection{Proof of Lemma \ref{cor.set}}
In order to prove Lemma \ref{cor.set}, we first state and proof a suitable
entropy bound:

We introduce the norm
$$\| h ( \cdot , \cdot) \|_{P_n} =
\sqrt { {1 \over n} \sum_{i=1}^n h^2 (x_i , Y_i)  } . $$
For a collection ${\cal H}$ of functions on ${\cal X} \times {\cal Y}$,
we let $H( \cdot , {\cal H} , \| \cdot \|_{P_n} ) $ be the entropy
of ${\cal H}$ equipped with the metric induced by the norm
$\| \cdot \|_{P_n}$ (for a definition of the entropy of a metric space see \cite{geer00}). 

Define for $ \epsilon >0 $, 
$$\tilde \Theta ( \epsilon) =
\{ \vartheta^T = ( \phi_1^T , \ldots , \phi_k^T , \eta^T )
 \in \tilde \Theta: \  \| \phi - \phi_0 \|_1 + 
 \ \| \eta - \eta_0 \|_2 \le \epsilon \} . $$

{\em Entropy Lemma} {\it For a constant $C_0$ depending on $k$ and $m$,
 we have for all $u>0 $ and $M_n>0$,
\begin{eqnarray*}
H\biggl  ( u ,\biggl  \{ ( L_{\vartheta}- L_{\vartheta^*} ) {\rm l}
 \{ G_1 \le M_n \}  :\ \vartheta \in \tilde \Theta(\epsilon) \biggr \} ,
 \| \cdot \|_{P_n} \biggr ) \le 
 C_0 { \epsilon^2 M_n^2 \over u^2 } \log \left ( {\epsilon M_n  \over u }
 \right ) .
\end{eqnarray*}
}
 
\emph{Proof (Entropy Lemma)} We have
\begin{eqnarray*}
|L_{\vartheta} (x,y) - L_{\tilde \vartheta} (x,y) |^2 &\le&   G_1^2 (y) 
   \biggl [  \sum_{r=1}^k   |(\phi_r - \tilde \phi_r)^T x |+ 
  \| \eta - \tilde \eta \|_1\biggr  ]^2\\ &\le& d G_1^2 (y)    \biggl [
  \sum_{r=1}^k   |(\phi_r - \tilde \phi_r)^T x |^2+  \| \eta - \tilde \eta
  \|_2^2\biggr  ].
\end{eqnarray*}
It follows that
\begin{eqnarray*}
\| ( L_{\vartheta}- L_{\tilde \vartheta}) {\rm l} \{ G_1 \le M_n \}
\|_{P_n}^2 \le 
  d M_n^2   \biggl [  \sum_{r=1}^k   {1 \over n} 
  \sum_{i=1}^n |(\phi_r - \tilde \phi_r)^T x_i |^2+ 
  \| \eta - \tilde \eta \|_2^2\biggr  ].
\end{eqnarray*}
Let $N( \cdot , \Gamma , \tau)$ denote the covering number of
a metric space $( \Gamma , \tau)$ with metric (induced by the norm) $\tau$, and
$H( \cdot , \Gamma , \tau)= \log N( \cdot , \Gamma , \tau)$ be its
entropy (for a definition of the covering number of a metric space see \cite{geer00}). If $\Gamma$ is a ball with radius $\epsilon$ in Euclidean space
$\R^N$, one easily verifies that
$$H(u , \Gamma, \tau) \le N \log \left({5\epsilon\over u}\right)  , \forall
\ u>0 . $$
Thus $H(u, \{\eta\in\R^m: ||\eta-\eta_{0}||_2\leq \epsilon\} , \| \cdot
\|_2 ) \le m\log \left({5\epsilon\over u}\right), \, \forall u>0.$
Moreover, applying a bound as in Lemma 2.6.11 of \cite{vaart96weak} gives
$$H\biggl( 2u  , \biggl\{ \sum_{r=1}^k(\phi_r-\phi_{0,r})^Tx_r: \| \phi-\phi_0\|_1\leq \epsilon \biggr\} , \| \cdot \|_{P_n}\biggr ) \le
\biggr({\epsilon^2 \over u^2}+1\biggl) \log (1+kp). $$
 
We can therefore conclude that
\begin{eqnarray*}
&&H\biggl ( 3\sqrt {d} M_n u , 
\biggl  \{ (L_{\vartheta}- L_{\vartheta_0}) {\rm l} \{ G_1 \le M_n \}  :\ \vartheta \in \tilde \Theta(\epsilon) \biggr \} ,
 \| \cdot \|_{P_n} \biggr) \\ &&\le\left (  {  \epsilon^2 \over u^2 } +m +1\right ) 
  \left ( \log \left({5\epsilon\over u}\right)+\log(1+kp) \right ) . 
\end{eqnarray*}

\medskip
Let's now turn to the main proof of Lemma \ref{cor.set}.

In what follows, $\{c_t\}$ are constants
depending on $\Lambda_{\max}$, $k$, $m$ and $K$. The truncated version of
the empirical process is 
$$
V_n^{\rm trunc} (\vartheta) = {1 \over n}
\sum_{i=1}^n \biggl ( L_{\vartheta} (x_i , Y_i)  {\rm l} \{ G_1(Y_i) \le M_n \} -
\EE \Bigl [ L_{\vartheta} (x_i , Y)  
 {\rm l} \{ G_1(Y) \le M_n \}  \Bigl \vert X = x_i \Bigr ] \biggr ). $$
Let $\epsilon>0$ be arbitrary.
We invoke Lemma 3.2  in \cite{geer00}, combined with 
a symmetrization lemma (e.g., a conditional version of Lemma 3.3 in \cite{geer00}). We apply these lemmas to the class
$$\biggl \{ (L_{\vartheta} - L_{\vartheta_0} ) {\rm l} \{ G_1 \le M_n \} : \ \vartheta \in
\tilde \Theta (\epsilon) \biggr \} . $$
In the notation used in Lemma 3.2 of \cite{geer00}, we take
\linebreak$\delta=c_4  T \epsilon M_n  \log n\sqrt{\log (p\vee n) / n }$, and $R = c_5 (\epsilon \wedge 1) M_n$.
This then gives
\begin{eqnarray*}
&&\PP_{\bf x} \biggl ( \sup_{\vartheta \in \tilde \Theta(\epsilon)} 
|  V_n^{\rm trunc} (\vartheta) -  V_n^{\rm trunc} ( \vartheta_0) | \ge c_6
\epsilon T M_n  \log n\sqrt { \log(p\vee n) \over  n }  \biggr ) \\
&&\le c_7 \exp\biggl [ - {T^2 \log^2 n \log(p\vee n) ( \epsilon^2 \vee 1) \over c_8^2 }\biggr  ] .
\end{eqnarray*}
Here, we use the bound (for $0<a\le 1$),
$$\int_{a}^1  { 1 \over u} \sqrt{ \log \biggl ( {1 \over u} } \biggr ) du \le
\log^{3/2} \biggl ( {1 \over a} \biggr )  . $$

We then invoke the peeling device: split the set $\tilde \Theta$ into sets
$$\{ \vartheta \in \tilde \Theta: \ 2^{-{(j+1)} }\le \| \phi - \phi_0 \|_1 + \| \eta - \eta_0 \|_2 \le 2^{-j} 
\} ,$$
where $j \in \Z$, and $2^{-{j+1}} \ge \lambda_0 $.
There are no more than $c_9 \log n $ indices $j\le 0$
with $2^{-{j+1}} \ge \lambda_0$. Hence, we get
$$\sup_{\vartheta^T = ( \phi^T , \eta^T)  \in \tilde \Theta } { 
 \biggl |  V_n^{\rm trunc} ( \vartheta ) -  V_n^{\rm trunc} ( \vartheta_0) 
\biggr | \over (
\| \phi - \phi^* \|_1 + \| \eta - \eta^* \|_2 ) \vee \lambda_0 } \le 2 c_6 TM_n  
\log n\sqrt { \log(p\vee n) \over n }   ,$$
with $\PP_{\bf x} $ probability at least
\begin{eqnarray*}
&&1- c_7 [c_9 \log n ] \exp\biggl  [ - {  T^2 \log^2 n \log(p\vee n)   \over c_8^2 } \biggr ] -
\sum_{j=1}^{\infty} c_7\exp\biggl  [ - {  T^2 2^{2j}  \log^2 n \log(p\vee n)   \over c_8^2 } \biggr ] \\
&&\ge 1 - c_2  \exp\biggl  [ - {  T^2 \log^2 n \log(p\vee n)   \over c_{10}^2 } \biggr ].
\end{eqnarray*}

Finally, to remove the truncation, we use
$$ | (L_{\vartheta} (x,y) - L_{\vartheta_0} (x,y) ){\rm l} \{ G_1(y) > M_n \} | \le
dK G_1(y) {\rm l} \{ G_1 (y) > M_n \} . $$
Hence
$${ \biggl |  (V_n^{\rm trunc} ( \vartheta ) -  V_n^{\rm trunc} ( \vartheta_0) )
- ( V_n (\vartheta) - V_n (\vartheta_0)) 
\biggr | \over (
\| \phi - \phi^* \|_1 + \| \eta - \eta^* \|_2 ) \vee \lambda_0 } $$
$$ \le {dK   \over n\lambda_0 } \sum_{i=1}^n \biggl ( G_1(Y_i) {\rm l} \{ G_1(Y_i) > M_n \} +
\EE \Bigl [ G_1(Y) {\rm l} \{ G_1(Y) > M_n \} \Bigl \vert X=x_i \Bigr ]\biggr )  .$$\hfill $\sqcup \mkern -12mu \sqcap$

\subsection{Proof of Theorem \ref{th.oracle}}

On ${\cal T}$
$$\bar {\cal E} (\hat \psi \vert \psi_0) +
\lambda   \| \hat \phi \|_1 \le
T \lambda_0 \biggl [  ( \| \hat \phi - \phi_0 \|_1 + \|\hat  \eta - \eta_0 \|_2 
 ) \vee \lambda_0 
\biggr ] + \lambda \| \phi_0 \|_1 + \bar {\cal E} (\psi_0\vert \psi_0). $$
By Lemma \ref{marginlemma},
$$\bar {\cal E} (\hat \psi \vert \psi_0 ) \ge \| \hat \psi - \psi_0 \|_{Q_n}^2 / c_0^2 ,
 $$
 and $\bar {\cal E} (\psi_0 \vert \psi_0 ) = 0$. 

{\bf Case 1} Suppose that
$$\| \hat \phi - \phi_0 \|_1 + \| \hat \eta - \eta_0 \|_2 \le \lambda_0 . $$
Then we find
$$ \bar {\cal E} (\hat \psi \vert \psi_0 ) \le
T \lambda_0^2  + \lambda \| \hat \phi - \phi_0 \|_1 +
\bar {\cal E} (\psi_0 \vert \psi_0 ) $$
$$ \le (\lambda + T \lambda_0) \lambda_0.$$

{\bf Case 2} Suppose that
$$\| \hat \phi - \phi_0 \|_1 + \| \hat \eta - \eta_0 \|_2 \ge \lambda_0 , $$
and that
$$T \lambda_0 \| \hat \eta - \eta_0 \|_2 
\ge ( \lambda + T \lambda_0 ) \| \hat \phi_S - (\phi_0)_S \|_1 . $$ 
Then we get
\begin{eqnarray*}
& &\bar {\cal E} (\hat \psi \vert \psi_0 ) + ( \lambda - T \lambda_0)
\| \hat \phi_{S^c} \|_1 \le 2 T \lambda_0 \| \hat \eta - \eta_0 \|_2 \\
&&\le 4T^2 \lambda_0^2  c_0^2+ \| \hat \eta - \eta_0 \|_2^2 / (2c_0^2) \\
&&\le 4T^2 \lambda_0^2 c_0^2 +  \bar {\cal E} (\hat \psi \vert \psi_0
)/2.
\end{eqnarray*}
So then
$$ \bar {\cal E} (\hat \psi \vert \psi_0 ) + 2 ( \lambda - T \lambda_0)
\| \hat \phi_{S^c} \|_1 \le 8T^2 \lambda_0^2 c_0^2.$$
  
  {\bf Case 3} Suppose that
$$\| \hat \phi - \phi_0 \|_1 + \| \hat \eta - \eta_0 \|_2 \ge \lambda_0 , $$
and that
$$T \lambda_0 \| \hat \eta - \eta_0 \|_2 
\le ( \lambda + T \lambda_0 ) \| \hat \phi_S - (\phi_0)_S \|_1 . $$
Then we have
$$ \bar {\cal E} (\hat \psi \vert \psi_0 ) +  ( \lambda - T \lambda_0)
\| \hat \phi_{S^c} \|_1 \le 2 (\lambda + T \lambda_0 ) \| \hat \phi_S -
\phi_0 \|_1 . $$
So then
$$\| \hat \phi_{S^c} \|_1 \le 6 \| \hat \phi_S - (\phi_0)_S \|_1 . $$
We can then apply the restricted
  eigenvalue condition to
$\hat \phi - \phi_0 $. This gives
\begin{eqnarray*}
\bar {\cal E} (\hat \psi \vert \psi_0 ) +  ( \lambda - T \lambda_0)
\| \hat \phi_{S^c} \|_1 &\le& 2 (\lambda + T \lambda_0 )
\sqrt s  \| \hat \phi_S -
\phi_0 \|_2 . \\
&\le& 2 (\lambda + T \lambda_0 )
\sqrt s  \kappa \| \hat g -
g_0 \|_{Q_n} \\
&\le& 4 (\lambda + T \lambda_0 )^2 c_0^2 \kappa^2 s + 
 \bar {\cal E} (\hat \psi \vert \psi_0 )/2.
\end{eqnarray*}
 So we arrive at
 $$ \bar {\cal E} (\hat \psi \vert \psi_0 ) +  2 ( \lambda - T \lambda_0)
\| \hat \phi_{S^c} \|_1\le 8 (\lambda + T \lambda_0 )^2 c_0^2 \kappa^2 s.$$\hfill $\sqcup \mkern -12mu \sqcap$

\subsection{Proof of Lemma \ref{lemm.setT}}
Let $Z$ be a standard normal random variable.
Then by straightforward computations, for all $M>0$,
$$\E[|Z| {\rm l} \{ |Z| > M \}] \le 2 \exp[-M^2 / 2] , 
$$ and
$$\E[|Z|^2 {\rm l} \{ |Z| > M \}] \le (M+2) \exp [-M^2 / 2 ] . $$
Thus, for $n$ independent copies $Z_1 , \ldots , Z_n$ of $Z$,
and $M = 2 \sqrt{\log n}$,
\begin{eqnarray*} 
&&\PP \left ( {1 \over n } \sum_{i=1}^n | Z_i | {\rm l} \{ | Z_i  | > M \} >
{ 4 \log n \over n }   \right )\\
&&\le \PP \left ( {1 \over n } \sum_{i=1}^n | Z_i | {\rm l} \{ | Z_i  | > M \}-
\E[|Z| {\rm l} \{ |Z| > M \}] > { 2{ \log n}  \over n } \right ) \\
&&\le { n \E[|Z|^2 {\rm l } \{ |Z| > M \}] \over4 (\log n)^2 } \le { 2 \over
  n }.
\end{eqnarray*} 
The result follows from this, as
$$G_1 (Y) = {\rm e}^K |Y| + K , $$
and $Y$ has a normal mixture distribution.\hfill $\sqcup \mkern -12mu \sqcap$

\subsection{Proof of Theorem \ref{theorem:consist-higd}}
On ${\cal T}$, defined in (\ref{setT}) with $\lambda_0 = c_4
\sqrt{\log^3n \log(p\vee n)/n}$ ($c_4$ as in Lemma \ref{lemm.setT}; i.e., $M_n = c_4
\sqrt{\log(n)}$ in (\ref{add1})), we have the basic inequality
$$\bar {\cal E} (\hat \psi \vert \psi_0) +
\lambda   \| \hat \phi \|_1 \le
T \lambda_0 \biggl [  ( \| \hat \phi - \phi_0 \|_1 + \|\hat  \eta - \eta_0 \|_2 
 ) \vee \lambda_0 
\biggr ] + \lambda \| \phi_0 \|_1 + \bar {\cal E} (\psi_0\vert \psi_0). $$
Note that $\|\hat{\eta} - \eta_0\|_2 \le 2K$ and $\bar {\cal E}
(\psi_0\vert \psi_0) = 0$. Hence, for $n$ sufficiently
large,
\begin{eqnarray*}
\bar {\cal E} (\hat \psi \vert \psi_0) +
\lambda   \| \hat \phi \|_1 &\le &
T \lambda_0 ( \| \hat \phi - \phi_0 \|_1 + 2K) + \lambda \| \phi_0 \|_1 +
\bar {\cal E} (\psi_0\vert \psi_0)\\ 
&\le & T \lambda_0 (\|\hat{\phi}\|_1 + \|\phi_0\|_1 + 2K) + \lambda \|
\phi_0 \|_1 + \bar {\cal E} (\psi_0\vert \psi_0),
\end{eqnarray*}
and therefore also 
\begin{eqnarray*}
\bar {\cal E} (\hat \psi \vert \psi_0) +
(\lambda - T \lambda_0)  \| \hat \phi \|_1 \le
T \lambda_0 2K + (\lambda + T \lambda_0)\| \phi_0 \|_1 + \bar {\cal E}
(\psi_0\vert \psi_0).
\end{eqnarray*}
It holds that $\lambda \ge 2 T \lambda_0$ (since $\lambda = C
\sqrt{\log^3n \log(p\vee n)/n}$ for some $C>0$ sufficiently large), $\lambda_0 =
O(\sqrt{\log^3n \log(p\vee n)/n})$ and 
$\lambda = O(\sqrt{\log^3n \log(p\vee n)/n})$, and due to the assumption about
$\|\phi_0\|_1$ we obtain on the set ${\cal T}$ that $\bar {\cal E} (\hat
\psi \vert \psi_0) \to \bar {\cal E} (\psi_0 \vert \psi_0) = 0\ (n
\to \infty)$. Finally, the set ${\cal T}$ has large probability, as shown
by 
Lemma \ref{cor.set} and using Proposition \ref{prop.FMR} and Lemma
\ref{lemm.setT} for FMR models. \hfill $\sqcup \mkern -12mu \sqcap$

\section{Proofs for Sections \ref{sec:penregr} and \ref{sec.optim}}\label{app2} 
\subsection{Proof of Proposition \ref{prop:bounded}}
We restrict ourselves to a two class mixture with k = 2. Consider the function
$u(\xi)$ defined as
\begin{align}\label{eq:proof.bounded}
u(\xi)=&\exp(\ell^{(0)}_{\mathrm{pen}}(\xi))\nonumber\\
=&\prod_{i=1}^{n}\Bigg\{\Bigg(\frac{\pi_1}{\sqrt{2\pi}\sigma_{1}}\exp\Big(\frac{-(Y_{i}-X^T_{i}\beta_{1})^{2}}{2\sigma_{1}^{2}}\Big)+\frac{(1-\pi_1)}{\sqrt{2\pi}\sigma_{2}}\\
&\times\exp\Big(\frac{-(Y_{i}-X^T_{i}\beta_{2})^{2}}{2\sigma_{2}^{2}}\Big)\Bigg)\exp\Big(\frac{-\lambda\|\beta_{1}\|_{1}}{\sigma_{1}}\Big)\exp\Big(\frac{-\lambda\|\beta_{2}\|_{1}}{\sigma_{2}}\Big)\Bigg\}.\nonumber
\end{align}
We will show that $u(\xi)$ is bounded from above for
$\xi=(\beta_1,\beta_2,\sigma_{1},\sigma_{2},\pi_1)\in\Xi=\R^{2p}\times\R_{>0}^{2}\times[0,1]$. Then,
clearly, $-n^{-1}\ell^{(0)}_{\mathrm{pen}}(\theta)$ is bounded from below for
$\theta=(\phi_1,\phi_2,\rho_{1},\rho_{2},\pi_1)\in\Theta=\R^{2p}\times\R_{>0}^{2}\times(0,1)$.

The critical point for unboundedness is if we choose for an arbitrary
sample point $i \in \{1,\ldots,n\}$ a $\beta^{*}_{1}$ such that
$Y_{i}-X_{i}^T\beta^{*}_{1}=0$ and let $\sigma_{1}\rightarrow
0$. Without the penalty term
$\exp(-\lambda\frac{\|\beta^{*}_{1}\|_{1}}{\sigma_{1}})$ in
(\ref{eq:proof.bounded}) the function would tend to infinity as
$\sigma_{1}\rightarrow0$. But as $Y_{i}\neq0$ for all $i \in \{1,\ldots,n\}$,
$\beta^{*}_{1}$ cannot be zero, and therefore
$\exp(-\lambda\frac{\|\beta^{*}_{1}\|_{1}}{\sigma_{1}})$ forces
$u(\xi)$ to tend to 0 as $\sigma_{1}\rightarrow 0$.

Let us give a more formal proof for the boundedness of $u(\xi)$. Choose a small $0<\eps_{1}<\min\limits_i\{Y^{2}_{i}\}$ and $\eps_{2}>0$. As $Y_{i}\neq0$, $i=1,\ldots,
n$, there exists a small constant $m>0$ such that
\begin{equation}\label{eq:proof.bounded1} 0<\min\limits_i\{Y^{2}_{i}\}-\eps_{1}\leq(Y_{i}-X_{i}^T\beta_{1})^{2}\end{equation}
holds for all $i=1,\ldots, n$ as long as $\|\beta_{1}\|_{1}<m$, and 
\begin{equation}\label{eq:proof.bounded2}0<\min\limits_i\{Y^{2}_{i}\}-\eps_{1}\leq(Y_{i}-X_{i}^T\beta_{2})^{2}\end{equation} holds for all $i=1,\ldots, n$ as long as $\|\beta_{2}\|_{1}<m$.

Furthermore, there exists a small constant $\delta>0$ such that
\begin{equation}\label{eq:proof.bounded3}\frac{1}{\sigma_{1}}\exp\Big(-(\min\limits_i\{Y^{2}_{i}\}-\eps_{1})/2\sigma_{1}^{2}\Big)<\eps_{2}\;\textrm{and}\;\frac{1}{\sigma_{1}}\exp\bigg(\frac{-\lambda
    m}{\sigma_{1}}\bigg)<\eps_{2}\end{equation}
hold for all $0<\sigma_{1}<\delta$, and 
\begin{equation}\label{eq:proof.bounded4}\frac{1}{\sigma_{2}}\exp\Big(-(\min\limits_i\{Y^{2}_{i}\}-\eps_{1})/2\sigma_{2}^{2}\Big)<\eps_{2}\;\textrm{and}\;\frac{1}{\sigma_{2}}\exp\bigg(\frac{-\lambda
    m}{\sigma_{2}}\bigg)<\eps_{2}\end{equation}
hold for all $0<\sigma_{2}<\delta$.

Define the set
$K=\{(\beta_1,\beta_2,\sigma_{1},\sigma_{2},\pi_1)\in\Xi;\ \delta\leq\sigma_{1},\sigma_{2}\}$.
Now $u(\xi)$ is trivially bounded on $K$. From the
construction of $K$ and Equations (\ref{eq:proof.bounded1})-(\ref{eq:proof.bounded4}), we easily see that $u(\xi)$ is also bounded on $K^{c}$,
and therefore bounded on $\Xi$.\hfill $\sqcup \mkern -12mu \sqcap$

\subsection{Proof of Theorem \ref{th.BCD-GEM}}
The density of the complete data is given by
\[f_{\mathrm{c}}(Y,\Delta|\theta)=\prod_{i=1}^{n}\prod_{r=1}^{k}\pi_{r}^{\Delta_{i,r}}\left(\frac{\rho_{r}}{\sqrt{2\pi}}e^{-\frac{1}{2}(\rho_{r}Y_{i}-X_{i}^T\phi_{r})^{2}}\right)^{\Delta_{i,r}}
,\]
whereas the density of the observed data is given by
\[f_{\mathrm{obs}}(Y|\theta)=\prod_{i=1}^{n}\sum_{r=1}^{k}\pi_{r}\frac{\rho_{r}}{\sqrt{2\pi}}e^{-\frac{1}{2}(\rho_{r}Y_{i}-X_{i}^T\phi_{r})^{2}},
\]
\begin{eqnarray*}
&&\theta=(\rho_{1},\ldots,\rho_{k},\phi_{1,1},\phi_{1,2},\ldots,\phi_{k,p},\pi)\in\Theta,\quad \Theta=\R_{>0}^{k}\times\R^{kp}\times\Pi\subset\R^{k(p+2)-1}, 
\end{eqnarray*}
with
\begin{align*}
&\Pi= \{\pi=(\pi_{1},\ldots,\pi_{k-1}); \pi_r > 0\ \mbox{for}\ r=1,\ldots
,k-1\ \mbox{and} \ \sum_{r=1}^{k-1} \pi_r < 1\},\quad\pi_{k}=1-\sum_{r=1}^{k-1}\pi_{r}.
\end{align*}
Furthermore, the conditional density  of the complete data given the
observed data is given
by $k(Y,\Delta|Y,\theta)=f_{\mathrm{c}}(Y,\Delta|\theta)/f_{\mathrm{obs}}(Y|\theta)$. Then, the
penalized negative log-likelihood fulfills the equation
\begin{eqnarray}\label{eq:property1}
\nu_{\mathrm{pen}}(\theta)&=& -n^{-1} \ell^{(0)}_{\mathrm{pen},\lambda}(\theta)\nonumber\\
&=&-n^{-1} \log f_{\mathrm{obs}}(Y|\theta)+\lambda\sum_{r=1}^{k}\|\phi_{r}\|_{1}\\
&=&\mathop{Q_{\mathrm{pen}}}(\theta|\theta')-\mathop{H}(\theta|\theta') \nonumber
\end{eqnarray}
where $\mathop{Q_{\mathrm{pen}}}(\theta|\theta')=-n^{-1} \EE_{\theta'}[\log
f_{\mathrm{c}}(Y,\Delta|\theta)|Y]+\lambda\sum_{r=1}^{k}\|\phi_{r}\|_{1}$
(compare Section \ref{subsec:emalgmix}) and
$\mathop{H}(\theta|\theta')=-n^{-1} \EE_{\theta'}[\log k(Y,\Delta|Y,\theta)|Y]$.

By Jensen's inequality, we get the following important relationship
\begin{equation}\label{eq:property2}
\mathop{H}(\theta|\theta')\geq \mathop{H}(\theta'|\theta') \qquad\forall\quad
  \theta\in \Theta,
\end{equation}
see also \cite{wu}.
$\mathop{Q_{\mathrm{pen}}}(\theta|\theta')$ and
$\mathop{H}(\theta|\theta')$ are continuous functions in $\theta$ and
$\theta'$. If we think of them as functions of $\theta$ with fixed $\theta'$,
we write also $\mathop{Q_{\mathrm{pen},\theta'}}(\theta)$ and
$\mathop{H}_{\theta'}(\theta)$. Furthermore,
$\mathop{Q_{\mathrm{pen},\theta'}}(\theta)$ is a convex function of $\theta$ and
strictly convex in each coordinate of $\theta$. As a last preparation, we give a definition of a stationary point for
non-differentiable functions (see also \cite{tseng}):
\begin{definition}
Let $u$ be a function defined on an open set $U \subset
\R^{k(p+2)-1}$. A point $x \in U$ is called stationary if
$$u'(x;d)=\lim_{\alpha \downarrow 0} \frac{u(x+\alpha d)-u(x)}{\alpha}\geq0
\quad \forall d\in\R^{k(p+2)-1}.$$
\end{definition}

We are now ready to start with the proof which is inspired by
\cite{Bertsekas_nonlin}. We write $$\theta=(\theta_1,\ldots,\theta_D)=(\rho_1,\ldots,\rho_k,\phi_{1,1},\phi_{1,2},\ldots,\phi_{k,p},\pi),$$ where
$D=k+kp+1$ denotes the number of coordinates. Remark that the first $D-1$
coordinates are univariate, whereas $\theta_{D}=\pi$ is a ``block
coordinate'' of dimension $k-1$.

\begin{proof} 
Let $\theta^{m}=\theta^{(m)}$ be the sequence generated by the BCD-GEM
algorithm. We need to prove that for a converging subsequence
$\theta^{m_{j}}\rightarrow \bar{\theta}\in\Theta$, $\bar{\theta}$ is a
stationary point of $\nu_{\mathrm{pen}}(\theta)$. Taking directional derivatives of Equation
(\ref{eq:property1}) yields
\[\nu_{\mathrm{pen}}'(\bar{\theta};d)=\mathop{Q'_{\mathrm{pen},\bar{\theta}}}(\bar{\theta};d)-\langle\nabla\mathop{H_{\bar{\theta}}}(\bar{\theta}),d\rangle.\]
Note that $\nabla\mathop{H_{\bar{\theta}}}(\bar{\theta})=0$ as
$\mathop{H}_{\bar{\theta}}(x)$ is minimized for $x = \bar{\theta}$
(Equation (\ref{eq:property2})). Therefore, it remains to show that
$\mathop{Q'_{\mathrm{pen},\bar{\theta}}}(\bar{\theta};d)\geq 0$ for all directions $d$.
Let
\[z_{i}^{m}=(\theta_{1}^{m+1},\ldots,\theta_{i}^{m+1},\theta_{i+1}^{m},\ldots,\theta_{D}^{m}).\]
Using the definition of the algorithm, we have
\begin{equation}\label{eq:alg}
\mathop{Q_{\mathrm{pen},\theta^{m}}}(\theta^{m})\geq \mathop{Q_{\mathrm{pen},\theta^{m}}}(z_{1}^{m})\geq \cdots \geq
\mathop{Q_{\mathrm{pen},\theta^{m}}}(z_{D-1}^{m})\geq \mathop{Q_{\mathrm{pen},\theta^{m}}}(\theta^{m+1}).
\end{equation}
Additionally, from the properties of GEM (Equation (\ref{eq:property1}) and (\ref{eq:property2})), we have
\begin{equation}\label{eq:monoton}
\nu_{\mathrm{pen}}(\theta^{0})\geq \nu_{\mathrm{pen}}(\theta^{1})\geq\ldots\geq \nu_{\mathrm{pen}}(\theta^{m})\geq \nu_{\mathrm{pen}}(\theta^{m+1}).
\end{equation}
Equation (\ref{eq:monoton}) and the converging subsequence imply that the sequence $\{\nu_{\mathrm{pen}}(\theta^{m});m=0,1,2,\ldots\}$
converges to $\nu_{\mathrm{pen}}(\bar{\theta})$. Further, we have
\begin{eqnarray}\label{eq:key}
0&\leq&\mathop{Q_{\mathrm{pen},\theta^{m}}}(\theta^{m})-\mathop{Q_{\mathrm{pen},\theta^{m}}}(\theta^{m+1})\nonumber\\
&&\nonumber\\
&=&\nu_{\mathrm{pen}}(\theta^{m})-\nu_{\mathrm{pen}}(\theta^{m+1})+\underbrace{\mathop{H_{\theta^{m}}}(\theta^{m})-\mathop{H_{\theta^{m}}}(\theta^{m+1})}_{\leq
  0}\nonumber\\
&\leq&\underbrace{\nu_{\mathrm{pen}}(\theta^{m})-\nu_{\mathrm{pen}}(\theta^{m+1})}_{\rightarrow
  \nu_{\mathrm{pen}}(\bar{\theta})-\nu_{\mathrm{pen}}(\bar{\theta})=0}.
\end{eqnarray}
We conclude that the sequence
$\{\mathop{Q_{\mathrm{pen},\theta^{m}}}(\theta^{m})-\mathop{Q_{\mathrm{pen},\theta^{m}}}(\theta^{m+1});m=0,1,2,\ldots\}$
converges to zero. 

We now show that $\{\theta_{1}^{m_{j}+1}-\theta_{1}^{m_{j}}\}$ converges to
zero ($j\to\infty$). Assume the contrary, in particular that $\{z_{1}^{m_{j}}-\theta^{m_{j}}\}$
does not converge to~0. Let
$\gamma^{m_{j}}=\|z_{1}^{m_{j}}-\theta^{m_{j}}\|$. Without loss of
generality (by restricting to a subsequence), we may assume that there exists
some $\bar{\gamma}>0$ such that $\gamma^{m_{j}}>\bar{\gamma}$ for all $j$.
Let
$s_{1}^{m_{j}}=\frac{z_{1}^{m_{j}}-\theta^{m_{j}}}{\gamma^{m_{j}}}$. This $s_{1}^{m_{j}}$ differs from zero only along the first component. As $s_{1}^{m_{j}}$ belongs to a compact set ($\|s_{1}^{m_{j}}\|=1$) we may assume that $s_{1}^{m_{j}}$ converges to $\bar{s}_{1}$.
Let us fix some $\eps \in [0,1]$. Notice that $0\leq \eps
\bar{\gamma}\leq \gamma^{m_{j}}$. Therefore, $\theta^{m_{j}}+
\eps\bar{\gamma} s_{1}^{m_{j}}$ lies on the segment joining $\theta^{m_{j}}$
and $z_{1}^{m_{j}}$, and belongs to $\Theta$ because $\Theta$ is convex. As
$\mathop{Q}_{\mathrm{pen},\theta^{m_{j}}}(\cdot)$ is convex and $z_{1}^{m_{j}}$ minimizes this function over all values that differ from $\theta^{m_{j}}$ along the first coordinate, we obtain
\begin{eqnarray}\label{eq:alg2} 
\mathop{Q_{\mathrm{pen},\theta^{m_{j}}}}(z_{1}^{m_{j}})&=&\mathop{Q_{\mathrm{pen},\theta^{m_{j}}}}(\theta^{m_{j}}+
\gamma^{m_{j}} s_{1}^{m_{j}})\nonumber\\
&\leq&\mathop{Q_{\mathrm{pen},\theta^{m_{j}}}}(\theta^{m_{j}}+\eps \bar{\gamma}
s_{1}^{m_{j}})\\
&\leq& \mathop{Q_{\mathrm{pen},\theta^{m_{j}}}}(\theta^{m_{j}}).\nonumber
\end{eqnarray}
From Equation (\ref{eq:alg}) and (\ref{eq:alg2}), we conclude
\begin{eqnarray*}
0&\leq& \mathop{Q_{\mathrm{pen},\theta^{m_{j}}}}(\theta^{m_{j}})-\mathop{Q_{\mathrm{pen},\theta^{m_{j}}}}(\theta^{m_{j}}+\eps
\bar{\gamma}s_{1}^{m_{j}})\\
&\stackrel{(\ref{eq:alg2})}{\leq}&\mathop{Q_{\mathrm{pen},\theta^{m_{j}}}}(\theta^{m_{j}})-\mathop{Q_{\mathrm{pen},\theta^{m_{j}}}}(z_{1}^{m_{j}})\\
&\stackrel{(\ref{eq:alg})}{\leq}& \mathop{Q_{\mathrm{pen},\theta^{m_{j}}}}(\theta^{m_{j}})-\mathop{Q_{\mathrm{pen},\theta^{m_{j}}}}(\theta^{m_{j}+1}).
\end{eqnarray*}
Using (\ref{eq:key}) and continuity of $\mathop{Q_{\mathrm{pen},x}}(y)$ in both
arguments $x$ and $y$, we conclude by taking the limit $j\rightarrow \infty$:
\begin{eqnarray*}
\mathop{Q_{\mathrm{pen},\bar{\theta}}}(\bar{\theta}+\eps \bar{\gamma}
\bar{s}_{1})&=&\mathop{Q_{\mathrm{pen},\bar{\theta}}}(\bar{\theta})\quad \forall
\eps \in [0,1].\\
\end{eqnarray*}
Since $\bar{\gamma}\bar{s}_{1}\neq 0$ this contradicts the strict convexity
of
$\mathop{Q_{\mathrm{pen},\bar{\theta}}}(x_{1},\bar{\theta}_{2},\ldots,\bar{\theta}_{D})$
as a function of the first block-coordinate. This contradiction establishes
that $z_{1}^{m_{j}}$ converges to $\bar{\theta}$.

From the definition of the algorithm, we have
\[
\mathop{Q_{\mathrm{pen}}}(z_{1}^{m_{j}}|\theta^{m_{j}})\leq
\mathop{Q_{\mathrm{pen}}}(x_{1},\theta_{2}^{m_{j}},\ldots,\theta_{D}^{m_{j}}|\theta^{m_{j}})\qquad
\forall x_{1}.
\]
By continuity and taking the limit $j\rightarrow \infty$, we obtain
\[
\mathop{Q_{\mathrm{pen},\bar{\theta}}}(\bar{\theta})\leq
\mathop{Q_{\mathrm{pen},\bar{\theta}}}(x_{1},\bar{\theta}_{2},\ldots,\bar{\theta}_{D})\qquad
\forall x_{1}.
\]
Repeating the argument we conclude that $\bar{\theta}$ is a coordinate-wise
minimum. Therefore, following \cite{tseng}, $\bar{\theta}$ is easily seen to be a stationary point of
$\mathop{Q_{\mathrm{pen},\bar{\theta}}}(\cdot)$, in particular
$\mathop{Q'_{\mathrm{pen},\bar{\theta}}}(\bar{\theta};d)\geq 0$ for all directions~$d$.
\end{proof}

\end{appendices}

\vspace{0.5cm}
{\bf Acknowledgements} {N.S. acknowledges financial support from Novartis International AG, Basel, Switzerland.}

\end{document}